\documentclass[journal,draftcls,onecolumn,12pt,twoside]{IEEEtranTCOM}
\normalsize

\usepackage{hyperref}
\usepackage{hypcap}
\usepackage{mathdefs}
\usepackage{graphicx}
\usepackage{epstopdf} 
\usepackage{amsthm}
\usepackage{amsmath}
\usepackage{amsfonts}
\usepackage{amssymb}
\usepackage{subfig} 
\usepackage{pbox}
\usepackage{wrapfig}
\usepackage{tikz}
\usetikzlibrary{shapes,backgrounds}
\usetikzlibrary{automata, positioning}
\usepackage{blkarray}
\usepackage{multirow}
\usepackage{capt-of}
\usepackage{etoolbox}
\usepackage{float}
\usepackage{cite}
\usepackage[skip=0pt]{caption}
\newcommand{\subparagraph}{}
\usepackage{titlesec}
\titlespacing{\section}{0pt}{\parskip}{-\parskip}
\titlespacing{\subsection}{0pt}{\parskip}{-\parskip}
\titlespacing{\subsubsection}{0pt}{\parskip}{-\parskip}
\newtheoremstyle{exampstyle}
{\topsep} 
{\topsep} 
{} 
{} 
{\bfseries} 
{.} 
{.3em} 
{} 

\tikzset{every picture/.style={font issue=\small},
	font issue/.style={execute at begin picture={#1\selectfont}}
}
\newcommand{\RN}[1]{%
	\textup{\uppercase\expandafter{\romannumeral#1}}%
}

\newcommand{\secref}[1]{\S\ref{sec:#1}}
\renewcommand{\eqref}[1]{(\ref{eq:#1})}
\newcommand{\figref}[1]{Fig.~\ref{fig:#1}}
\newcommand{\tabref}[1]{Table~\ref{tab:#1}}

\newcommand{\thmref}[1]{Theorem.~\ref{thm:#1}}
\newcommand{\lemref}[1]{Lemma.~\ref{lem:#1}}
\newcommand{\propref}[1]{Proposition.~\ref{prop:#1}}
\newcommand{\defref}[1]{Definition.~\ref{def:#1}}
\newcommand{\zerosf}{\mathsf{0}}
\newcommand{\onesf}{\mathsf{1}}

\theoremstyle{exampstyle} 
\theoremstyle{exampstyle} 
\theoremstyle{exampstyle} \newtheorem{definition}{Definition}
\theoremstyle{exampstyle} 
\theoremstyle{exampstyle} 
\theoremstyle{exampstyle} 
\theoremstyle{exampstyle} \newtheorem{lemma}{Lemma}
\theoremstyle{exampstyle} \newtheorem{prop}{Proposition}
\theoremstyle{exampstyle} \newtheorem{theorem}{Theorem}
\theoremstyle{exampstyle} 
\graphicspath{{./}}

\hypersetup{
	pdftitle={Your title here},
	pdfauthor={Your name here},
	pdfsubject={Your subject here},
	pdfkeywords={keyword1, keyword2},
	bookmarksnumbered=true,     
	bookmarksopen=true,         
	bookmarksopenlevel=1,       
	colorlinks=true,            
	pdfstartview=Fit,           
	pdfpagemode=UseOutlines,    
	pdfpagelayout=TwoPageRight
}

\title{Efficiency and detectability of random reactive jamming in carrier sense wireless networks}

\author{Ni~An,~\IEEEmembership{Member,~IEEE,} and~Steven~Weber,~\IEEEmembership{Senior Member,~IEEE}
\thanks{N. An and S. Weber are with the Department of Electrical and Computer Engineering, Drexel University, Philadelphia, PA. This research has been supported by the National Science Foundation under award \#CNS-1228847. Preliminary version of this work was presented at the 15th Annual IEEE International Conference on Sensing, Communication, and Networking (SECON) in June, 2018 in Hong Kong, China \cite{AnWeb2018}. S. Weber is the contact author ({\sf sweber@coe.drexel.edu}).}}

\begin{document}
\IEEEoverridecommandlockouts
\maketitle

\begin{abstract}
A natural basis for the detection of a wireless random reactive jammer (RRJ) is the perceived violation by the detector (typically located at the access point (AP)) of the carrier sensing protocol underpinning many wireless random access protocols ({\em e.g.}, WiFi).  Specifically, when the wireless medium is perceived by a station to be busy, a carrier sensing compliant station will avoid transmission while a RRJ station will often initiate transmission.  However, hidden terminals (HTs), {\em i.e.}, activity detected by the AP but not by the sensing station, complicate the use of carrier sensing as the basis for RRJ detection since they provide plausible deniability to a station suspected of being an RRJ.  The RRJ has the dual objectives of avoiding detection and effectively disrupting communication, but there is an inherent performance tradeoff between these two objectives.  In this paper we capture the behavior of both the RRJ and the compliant stations via a parsimonious Markov chain model, and pose the detection problem using the framework of Markov chain hypothesis testing.  Our analysis yields the receiver operating characteristic (ROC) of the detector, and the optimized behavior of the RRJ.  While there has been extensive work in the literature on jamming detection, our innovation lies in leveraging carrier sensing as a natural and effective basis for detection.
\end{abstract}


\section{Introduction}

Jamming attacks are a widely recognized threat to wireless networks. As a type of {\em denial-of-service} attack, wireless jamming leverages the broadcast nature of the wireless medium and emits jamming signals either to prevent other (compliant) users from accessing the network, or to corrupt ongoing transmissions. There are three major types of jammers \cite{XuTra2005}: $i)$ {\em constant jammer}, which constantly sends jamming signals, $ii)$ {\em random jammer}, which randomly alternates between jamming and idle states, and $iii)$ {\em reactive jammer}, which emits jamming signals upon sensing any ongoing traffic over the wireless channel. Compared with the first two types, the reactive jammer (RJ) is more sophisticated in that it achieves high jamming efficiency by only disrupting ongoing transmissions, which in general also lowers the risk of detection \cite{XuTra2005}.  A RJ faces an inherent tradeoff in the dual objectives of effectively degrading network throughput and in avoiding detection: as the ``aggressiveness'' of the jamming is increased, it increases the effectiveness of the disruption, but at the same time increases the ease with which behavior not compliant with carrier sensing is detected.  This detection is often based upon changes in network performance statistics such as the packet delivery rate (PDR), received signal strength (RSS), packet delivery delay, {\em etc}. However, the presence of {\em hidden terminals (HT)}, {\em i.e.}, transmissions detectable by the access point (AP) but not the sensing station, complicates the jamming detection problem, as the AP cannot always disambiguate whether a new packet is a (malicious) jamming decision or a (innocuous) HT mistake \cite{MokBen2015}.  This motivates our work on RJ detection and RJ design in the presence of HTs. There are several related detection problems: $i)$ deciding whether or not a specified (suspicious) station is a jammer, which is the focus of this paper, $ii)$ identifying which station is the jammer given knowledge that there is a jammer in the network, and $iii)$ deciding whether or not each station in the network is a jammer or a compliant station.

\subsection{Related works}

The detection of general jamming attacks has been extensively studied in \cite{XuTra2005,MokBen2015,ShiShe2009,LiKou2007,PelIli2011,LuWan2014,MarWyg2014,PunAkt2014,SpuGiu2014, CiuAub2017}. Xu \emph{et al.} \cite{XuTra2005} analyze the influence of various jamming attacks on the PDR and RSS of the network, and propose a thresholding algorithm for jamming detection. Other works such as \cite{MarWyg2014, PunAkt2014} utilize different metrics, such as the channel busy ratio, the number of retransmission attempts, etc., in addition to the metrics proposed by Xu \emph{et al.}, and employ machine learning based techniques for jamming detection. Shin \emph{et al.} \cite{ShiShe2009} propose an approach based on group testing to identify the trigger stations, whose signal triggers the RJ activity, in wireless sensor networks. Lu \emph{et al.} \cite{LuWan2014} investigates jamming attacks in time-critical networks and present analytical results of the network message invalidation ratio under jamming. There is also a body of work on analyzing jamming attacks' effects on the performance of wireless networks \cite{HarPin2011,BenBou2011,BayKin2013}. Bayraktaroglu \emph{et al.} \cite{BayKin2013} present theoretical results of the IEEE 802.11 throughput under various jamming attacks, and their analysis is mainly built upon Bianchi's Markov chain model of 802.11 DCF \cite{Bia2000}. 

Although there is a large body of work on jamming attacks, few analytical results have been developed about the relationship between the {\em effectiveness} and the {\em detectability} of jamming attacks, which is the central focus of our work. Li \emph{et al.} consider mathematical models of an optimal jamming attack which chooses its jamming probability that balances the tradeoff between the long-term amount of corrupted packets and the detection time under the slotted Aloha protocol \cite{LiKou2007}. The jamming detection algorithm that Li \emph{et al.} employed is a sequential probability ratio test based on the amount of collision events. One drawback of Li \emph{et al.}'s work is that it only considers the slotted Aloha protocol, which does not incorporate carrier sense multiple access (CSMA), an essential feature of the ubiquitous IEEE 802.11 protocol. 

\subsection{Contributions and outline}

This paper focuses on both RJ attack design and the detection of RJ attacks by the AP, in the presence of HTs. We consider a ``single-hop'' wireless network in which multiple wireless stations communicate directly with a single AP, equipped with a jamming detection monitor.  Our work is distinct from previous work such as \cite{LiKou2007} in that our detector leverages the CSMA mechanism underlying many modern wireless multiple access protocols.  Wireless stations are {\em compliant} in the sense that they are assumed to comply with the CSMA mechanism, meaning that they sense the wireless channel and only transmit if and when the medium is sensed as {\em idle}, while reactive jammers are bad actors that violate the CSMA mechanism, meaning that they sense the channel and only transmit when the medium is sensed as {\em active}.  This behavior, when identified, differentiates the reactive jammers from the compliant stations, and is the basis for the AP's ability to detect RJ.  The difficulty of this detection, however, is that the AP cannot disambiguate whether a transmission on top of an active channel is attributable to the (innocuous) HT or to the (malicious) reactive jammer.

\begin{table}[h]
\centering
\begin{tabular}{c|l}
\hline
Symbol 		& Meaning\\
\hline
$m$ & number of stations\\
$T_k$ & a bit indicating the idle/active status of the $k$-th station\\
$\Smc$ & the state space of the full observability model with size $(d+1)$\\
$d$ & the largest index of the state space of the Markov chain (indexing starts from 0)\\
$\Tmc$ & a state in $\Smc$, denoting a subset of nodes in $[m]$ that are active\\
$p_I(k, \Tmc)$ & probability that station $k$ senses the channel as idle with the set of active stations being $\Tmc$ \\
$p_A(1, \Tmc)$ & probability that the RRJ $1$ sends jam packets when the set of active stations being $\Tmc$ \\
$p_J$ & the RJ probability \\
$p_R$ & the random jamming probability \\
$\pbf_{\rm rrj}$ & a row vector equals to $[p_R, p_J]$ \\
$\lambda$ 	& a single station's off to on transition rate when probability of channel being sensed as idle is 1\\
$\gamma$    & a single station's on to off transition rate\\
$\Hsf_{\bsf}$ & null hypothesis when $\bsf=0$, alternative hypthesis when $\bsf=1$\\
$\Qbf^{\bsf}$ & transition rate matrix under $\Hsf_{\bsf}$\\
$q^{\bsf}_{i,j}$ & the transition rate from state $i$ to state $j$ under $\Hsf_{\bsf}$\\
$\Pbf^{\bsf}$ & transition matrix under hypothesis $\Hsf_{\bsf}$ \\
$p_{i,j}^{\bsf}$ & transition probability from state $i$ to state $j$ under $\Hsf_{\bsf}$\\
$\boldsymbol{\pi}^{\bsf}$ & a row vector denoting the stationary distribution of $\Pbf^{\bsf}$ (or $\Qbf^{\bsf}$)\\
$\pi_j^{\bsf}$ & the $j$-th element of $\boldsymbol{\pi}^{\bsf}$\\
$W$ & length of the sample path\\
$N_{i,j}$ & the number of transitions from state $i$ to state $j$ \\
$Z^{\bsf}$ & the log-likelihood ratio test statistic random variable \eqref{defZ} under $\Hsf_{\bsf}$\\
$\eta$ & jamming efficiency metric \eqref{etaDef}\\
\hline     
\end{tabular}
\caption{Notation.}
\label{tab:notation}
\end{table}

Our contributions include: $i)$ we design a novel algorithm based on a Markov chain model that detects RJ by its violation of the CSMA mechanism in a network with HT problems, $ii)$ we analyze the variance of the test statistic of a generic binary hypothesis testing of ergodic Markov chains, $iii)$ we build a novel intelligent RRJ model with the two competing goals to evade the jamming detector and to maximize its jamming efficiency, $iv)$ we propose and analyze the jamming detection model under different assumptions regarding the knowledge of the jammer and the AP regarding the state of the network. 

The rest of this paper is organized as follows. \secref{fullModel} formulates the basic mathematical model and sets up the hypothesis test for jamming detector with full observability, proposes an analytical upper bound on the test statistic's variance, and proposes the model for a jammer to choose its best jamming strategy. \secref{limitedModel} considers the cases when the detector only has limited observability. \secref{Experiments} shows numerical results, and \secref{conclusion} concludes the paper. \tabref{notation} lists general notation.

\section{Markov model for full observability}
\label{sec:fullModel}

One objective of this paper is to propose a tractable model that captures the essence of CSMA to detect potentially non-compliant stations, such as jammers. \secref{mathematicalModel} proposes a novel continuous-time Markov chain (CTMC) model for the overall transmission behavior of a network consisting of CSMA-compliant stations, and also proposes a more advanced jammer called a {\em random reactive jammer (RRJ)}. \secref{NPTestMC} introduces the general supervised hypothesis testing problem of Markov chain models. \secref{rrjStrategy} presents an approach for selecting the best jamming strategy for an intelligent RRJ. \secref{semiSupvised} introduces the semi-supervised testing problem for when only the behavior of compliant stations are available to the detector for training.
 
\subsection{Mathematical models}
\label{sec:mathematicalModel}

There have been extensively works on modeling the CSMA protocols in the literature: one group of works concentrates on using an idealized CTMC model for CSMA/CA WiFi networks and approximating their throughput \cite{LieKai2010,BorHil2013,NarKni2012,BelZoc2014,FohZuk2007,JiaWal2010,LauKle2013}; while \cite{BelChe2016,MicRog2016} focus on using CTMC to model the interaction between the APs of multiple networks to analyze their performance. The proposed CTMC in our paper is distinct from  previous works in that: $i)$ our model is more realistic in that it is built on the physical interference model of the channel, while previous works employ a simple contention-graph based flow model that assumes WiFi signal sensing is deterministic and does not take the randomness of interference into consideration;  $ii)$ the objectives of our work are distinct from previous works in that our work aims to use the CTMC model to infer channel statistics, such as the fraction of certain station's transmitting time, and also to detect jammers using the state transition statistics. 

This section assumes the network monitor, assumed to be located at the AP, is ``omniscient'' in the sense that it is aware of the transmission behavior of all stations in the network, which consists of $m$ stations and the AP.  For simplicity, we also assume the stations are immobile and all stations have the same transmission power $p_t$.  We assume all stations are backlogged, {\rm i.e.,} each station always has a packet awaiting transmission, and as such the status of a single station $k \in [m]$ may be represented by a bit: $T_k=0$ or 1 indicates station $k$ is {\em idle} or {\em active}.  The overall transmission behavior of the whole network is modeled by a CTMC with {\em state space} $\Smc$, with cardinality $|\Smc| = d+1 \equiv 2^m$, and each state $\Tmc$ representing a distinct subset of $[m]$.  The state of the system is the subset of active stations, {\rm i.e.,} $\Tmc=\{k \in [m]: T_k=1\}$. 
We define several network parameters: the received power $p_o$ at a reference distance $d_o$, the {\rm Rayleigh fading random variable} of the $k$-th station $F_k\sim \mbox{Exp}(1)$, the pathloss exponent $\alpha$, the location vector of all the stations in $[m]$, denoted $\xbf=[x_1,x_2,...,x_m]$, and the minimum received power required for a station to detection transmission is $\theta$.  The service rate is denoted by $\gamma$, the sensing rate is $\lambda$, and $p_I(k, \Tmc)$ denotes the probability the channel is sensed as {\em idle} at station $k$ when the set of current active stations is $\Tmc$:
\begin{equation}
\label{eq:idleProb}
p_I(k,\Tmc)\equiv \Pbb\left(\sum_{k'\in \Tmc}F_{k'}l(\lVert x_{k'}-x_{k}\rVert)+ N_0 \leq \theta\right).
\end{equation}
Here, $\lVert x_{k'}-x_{k}\rVert$ denotes the Euclidean distance between stations $k'$ and $k$, $N_0$ denotes the background noise power, and $l(d)$ (with $l(d) = p_o d^{-\alpha} \text{ if $d \geq d_o$}$, and $l(d) = p_t, \text{if $d< d_o$}$) is the large-scale pathloss model.  As such, \eqref{idleProb} gives the probability that station $k$ senses the medium to be idle (as the received power level is below the minimum power level threshold $\theta$ required for detection of activity) when the set of concurrent transmitters is $\Tmc$.  Under the assumption that the network topology is static, the idle probability $p_I(k, \Tmc)$ defined in \eqref{idleProb} is the cumulative distribution function (CDF) of a weighted sum of independent exponential random variables with weights $l(\lVert x_{k'}-x_k\rVert)$. Considering the possibility that some stations have equal distance to station $k$, we  partition every active station set $\Tmc$ to $s(\Tmc)\geq 1$ disjoint groups: $\tau_1,...,\tau_{s(\Tmc)}$, each of which contains stations with equal distance to $k$: $\tau_{a}\equiv \{k'\in \Tmc: l(\lVert x_{k'}-x_k\rVert)=d_a\}$ for $a\in \{1,...,s(\Tmc)\}$. Then the interference in \eqref{idleProb} can be expressed as:
\begin{equation}
\label{eq:idleProbEqualDist}
\sum_{k'\in\Tmc}F_{k'} l(\lVert x_{k'}-x_k\rVert) = \sum_{a=1}^{s(\Tmc)} d_a \sum_{k'\in \tau_{a}}F_{k'}.
\end{equation}
Applying Amari and Misra's derivation of the general CDF of summation of independent exponential random variables in \cite{AmaMis1997}, and combining \eqref{idleProbEqualDist}, we may compute \eqref{idleProb} as:
\begin{equation}
\Pbb\left[\left(\sum_{a=1}^{s(\Tmc)} d_a \sum_{k'\in \tau_{a}}F_{k'}\right) \leq \theta'\right] 
\!\!=\! 1-\left(\prod_{a=1}^{s(\Tmc)} d_a^{-|\tau_{a}|}\right)\sum_{l=1}^{s(\Tmc)}\sum_{j=1}^{|\tau_l|}\frac{\mbox{exp}(-d_l^{-1} \theta') \theta'^{|\tau_l|-j} \Psi_{l,j}(-d_l^{-1})}{(|\tau_l|-j)!(j-1)!},
\end{equation}
where $\theta'=\theta - N_0$, and $\Psi_{l,j}(d)=-\frac{\partial^{j-1}}{\partial d^{j-1}}\{\prod_{a=0, a\neq l}^{s(\Tmc)}(d_a^{-1}+d)^{-|\tau_a|}\}$ (note that at the special point with $a=0$, let $(d_0^{-1}+d)^{-|\tau_0|}=d^{-1}$).  These probabilities are used in the CTMC transition rates, described below.  We begin below with three variants on a simple $m=2$ station network, then generalize in \secref{generalFullModel} to a network with an arbitrary number $m$ of stations.   

\subsubsection{CTMC for $m=2$ stations}
\label{sec:twoStationFullModel}
\begin{figure}[!htb]
	\begin{minipage}[t]{0.3\textwidth}
		\centering
		\begin{tikzpicture}[font=\sffamily,scale=0.46]
		\tikzset{node style/.style={state, 
				minimum width=1cm,
				line width=0.3mm,
				fill=gray!20!white}}
		\node[node style] at (0, 0)     (00)     {$\phi$};
		\node[node style] at (0, 5.5)     (10) 	 {$\{1\}$};
		\node[node style] at (5.5, 0)     (01)     {$\{2\}$};
		\node[node style] at (5.5, 5.5)     (11)     {$\{1,2\}$};
		\draw[every loop,
		auto=right,
		line width=0.3mm,
		>=latex,
		draw=black,
		fill=black]
		(00)     edge[bend right=20]     node {$\lambda$} (10)
		(10)     edge[bend right=20, auto=right] node {$\gamma$} (00)
		(10)     edge[bend left=20, auto=left]     node {$\lambda p_I(2,\{1\})$} (11)
		(11)     edge[bend left=20, auto=left] node {$\gamma$} (10)
		(11) edge[bend right=20] node {$\gamma$} (01)
		(01) edge[bend right=20, auto=right] node {$\lambda p_I(1,\{2\})$} (11)
		(00) edge[bend left=20, auto=left] node {$\lambda$} (01)
		(01) edge[bend left=20, auto=left] node {$\gamma$} (00)
		;
		\end{tikzpicture}
		\caption{CTMC of compliant stations $m=2$).}
		\label{fig:compliantCTMCm2}
	\end{minipage}
	\hfill	
	\begin{minipage}[t]{0.3\textwidth}
		\centering
		\begin{tikzpicture}[font=\sffamily,scale=0.46]
		\tikzset{node style/.style={state, 
				minimum width=1cm,
				line width=0.3mm,
				fill=gray!20!white}}
		\node[node style] at (0, 0)     (00)     {$\phi$};
		\node[node style] at (0, 5.5)     (10) 	 {$\{1\}$};
		\node[node style] at (5.5, 0)     (01)     {$\{2\}$};
		\node[node style] at (5.5, 5.5)     (11)     {$\{1,2\}$};
		\draw[every loop,
		auto=right,
		line width=0.3mm,
		>=latex,
		draw=black,
		fill=black]
		(10)     edge[bend right=0, auto=right] node {$\gamma$} (00)
		(10)     edge[bend right=20]     node {$\lambda p_I(2,\{1\})$} (11)
		(11)     edge[bend right=20, auto=right] node {$\gamma$} (10)
		(11) edge[bend right=20] node {$\gamma$} (01)
		(01) edge[bend right=20, sloped, anchor=center, below] node {$\lambda p_J (1-p_I(1,\{2\}))$} (11)
		(00) edge[bend right=20] node {$\lambda $} (01)
		(01) edge[bend right=20, auto=right] node {$\gamma$} (00)
		;
		\end{tikzpicture}
		\caption{CTMC with one naive RJ ($m=2, k_J = 1$).}
		\label{fig:NaiveRJCTMCm2}
	\end{minipage}
	\hfill
	\begin{minipage}[t]{0.3\textwidth}
		\centering
		\begin{tikzpicture}[font=\sffamily,scale=0.46]
		\tikzset{node style/.style={state, 
				minimum width=1cm,
				line width=0.3mm,
				fill=gray!20!white}}
		\node[node style] at (0, 0)     (00)     {$\phi$};
		\node[node style] at (0, 5.5)     (10) 	 {$\{1\}$};
		\node[node style] at (5.5, 0)     (01)     {$\{2\}$};
		\node[node style] at (5.5, 5.5)     (11)     {$\{1,2\}$};
		\draw[every loop,
		auto=right,
		line width=0.3mm,
		>=latex,
		draw=black,
		fill=black]
		(00)     edge[bend right=20, auto=right]     node {$\lambda p_A(1,\phi)$} (10)
		(10)     edge[bend right=20, auto=left] node {$\gamma$} (00)
		(10)     edge[bend right=20, auto=right]     node {$\lambda p_I(2,\{1\})$} (11)
		(11)     edge[bend right=20, auto=right] node {$\gamma$} (10)
		(11) edge[bend right=20, auto=left] node {$\gamma$} (01)
		(01) edge[bend right=20, auto=right] node {$\lambda p_A(1,\{2\})$} (11)
		(00) edge[bend right=20] node {$\lambda$} (01)
		(01) edge[bend right=20, auto=right] node {$\gamma$} (00)
		;
		\end{tikzpicture}
		\caption{CTMC with one RRJ ($m=2, k_J=1$).}
		\label{fig:RRJCTMCm2}
	\end{minipage}
\end{figure}  
{\em Two compliant stations.}  A compliant station does not transmit if it senses the channel is idle, but will transmit (as it is assumed backlogged) if it senses the channel is busy. Therefore, the transition rate of station $k$ from idle to active is $\lambda p_I(k,\Tmc)$, where recall $\lambda$ is the sensing rate, as shown in \figref{compliantCTMCm2}. The HT phenomenon in the network can be clearly captured by the idle probability: if $\Tmc$ contains only stations that are HTs relative to station $k$, then the idle probability $p_I(k,\Tmc)$ is relatively high, and the transition rate into the HT state $\Tmc \cup \{k\}$ is thus high as well. 

{\em One compliant station and one naive RJ.} The behavioral difference between a compliant station, potentially a HT, and a naive RJ is that a HT will transmit regardless of its HT counterparts (since HT pairs cannot hear each other's signal), while a naive RJ only transmits with a certain {\em jamming probability} as soon as it senses signals from other compliant stations on the channel. Define the jamming probability as $p_J$. The probability that a naive RJ $k_J$ can sense the current signal over the channel is $p_I(k_J,\Tmc)$, and thus the transition rate of naive RJ from idle to active is: $\lambda p_J (1-p_I(k_J,\Tmc))$. The interactions between stations in a network with a naive RJ is shown in \figref{NaiveRJCTMCm2}.  As shown, the transition rate from state $\phi$ ({\rm i.e.,} there are no active stations) to state $\{k_J\}$ ({\em i.e.}, jammer $k_J$ starts transmitting) is zero, since the naive RJ is only ``triggered'' by an active channel.  The state transition graphs in \figref{NaiveRJCTMCm2} and  \figref{compliantCTMCm2} are different, and in this case the Neyman-Pearson test of differentiating these two CTMCs is {\em degenerate} to a singular detection problem \cite{Lev2008}, meaning that the test can achieve arbitrarily small error \cite{MarMat2006}. 

{\em One compliant station and one RRJ.} To avoid the singular detection problem, we propose the RRJ model, which equips the RJ with additional randomness. \figref{RRJCTMCm2} shows the CTMC for a network containing a RRJ. The RRJ can better disguise its malicious behavior by mimicking HT: namely, when there is no traffic over the channel, the RRJ remains idle or randomly transmits packets with probability $p_R$, and when there are compliant packets transmitting, the RRJ decides whether or not to jam the compliant packet with certain jamming probability $p_J$. The proposed RRJ model increases the detection difficulty since the incorporation of random jamming behavior makes the RRJ similar to a HT. Thus we can define the {\em anomalous probability} of the RRJ $k_J$ as the probability that the RRJ sends jam packets when the set of active stations is $\Tmc$
\begin{equation}
\label{eq:pA}
p_{A}(k_J, \Tmc)\equiv p_R p_I(k_J,\Tmc) + p_J (1-p_I(k_J,\Tmc)).
\end{equation}
The RRJ has $(p_R,p_J)$ as design parameters.  The transition rate of a RRJ from idle to active is
\begin{equation}
\label{eq:rrjTransitionRate}
q_{\Tmc, \Tmc\cup\{k_J\}} = \lambda p_A(k_J, \Tmc).
\end{equation} 

\subsubsection{CTMC for an arbitrary number of stations}
\label{sec:generalFullModel}

We now extend the CTMC to an arbitrary number of stations $m$. Without loss of generality, index the {\rm station under test (SUT)}, {\rm i.e.,} the station which the AP is assessing, as the first station ({\em i.e.}, $k_J=1$), and use indices $\{2,\ldots,m\}$ to denote the other stations, assumed to be CSMA-compliant (hereafter referred to as {\em compliant stations} (CS)).  \figref{compliantCTMCm} shows the state transition diagram of the CTMC for a network without a RRJ (SUT is CS), while \figref{RRJCTMCm} shows the same with a RRJ (SUT is RRJ). From \eqref{rrjTransitionRate}, the RRJ has two design parameters, $(p_J,p_R)$, with $p_J$ controlling its ``reactive jamming'' behavior, and $p_R$ controlling its ``random jamming'' behavior. Define $\Qbf^{\bsf}\in \Rbb^{(d+1)\times (d+1)}$ for $\bsf \in \{\zerosf, \onesf\}$, where $\Qbf^{\zerosf}, \Qbf^{\onesf}$ are the transition rate matrices for CTMCs without and with the RRJ, respectively, and $q^{\bsf}_{i,j}$ is the $(i,j)$-th element of $\Qbf^{\bsf}$, denoting the transition rate from state $i$ to state $j$ in $\Smc$. {Note that for notation simplifity, we use $i$ (or $j$) to interchangeably denote: $i)$ a specific state $i$ in $\Smc$; $ii)$ the specific index of state $i$ in $\Smc$ according to certain ordering of states in $\Smc$, in the rest of this paper}. From \figref{compliantCTMCm}, if $i=\Tmc$ and $j = \Tmc\cup \{k\}$, then $q^{\zerosf}_{i,j} = \lambda p_I(k, \Tmc)$ and $q^{\zerosf}_{j,i} = \gamma$. From \figref{RRJCTMCm}, we have $q^{\onesf}_{i,j} = \lambda p_I(k, \Tmc)$ (if $k\neq 1$), $q^{\onesf}_{i,j} = \lambda p_A(1, \Tmc)$, and $q^{\onesf}_{j,i} = \gamma$. 

\begin{figure}[!htb]
	\centering
	\begin{minipage}[b]{0.46\textwidth}
	\begin{tikzpicture}[font=\sffamily,scale=0.55]
	\tikzset{node style/.style={draw, 	
			ellipse, 
			align=center,
			minimum width=1cm, 
			line width=0.3mm,
			fill=gray!20!white}}
	\node[node style] at (0, 2.5)     (T)     {$\Tmc$};
	\node[node style, text width = 2cm] at (8, 2.5)     (TPlusI) 	 {$\Tmc \cup \{k\},$\\ $\forall k\in [m]\backslash\Tmc$};
	\draw[every loop,
	auto=right,
	line width=0.3mm,
	>=latex,
	draw=black,
	fill=black]
	(T)     edge[bend left=20, above]     node {$\lambda p_I(k,\Tmc)$} (TPlusI)
	(TPlusI) edge[bend left=20] node {$\gamma$} (T)
	;
	\end{tikzpicture}
	\caption{CTMC without reactive jammers ($m>1$).}
	\label{fig:compliantCTMCm}
	\end{minipage}
	\hfill
	\begin{minipage}[b]{0.46\textwidth}
	\begin{tikzpicture}[font=\sffamily,scale=0.55]
	\tikzset{node style/.style={draw, 	
			ellipse, 
			align=center,
			minimum width=1cm, 
			line width=0.3mm,
			fill=gray!20!white}}
	\node[node style] at (0, 3)     				  (T)        {$\Tmc$};
	\node[node style, text width = 3cm] at (7, 5)     (TPlusI) 	 {$\Tmc \cup \{k\},$\\ $\forall k\in [m]\backslash (\Tmc \cup \{1\})$};
	\node[node style, text width = 1.5cm] at (5.5, 0)     (TPlusK) 	 {$\Tmc \cup \{1\},$ if $1\notin \Tmc$};
	\draw[every loop,
	auto=right,
	line width=0.3mm,
	>=latex,
	draw=black,
	fill=black]
	(T)     edge[bend left=20, above]     node {$\lambda p_I(k,\Tmc)$} (TPlusI)
	(T) edge[bend left=20, below] node {$ \quad\quad\quad\quad\quad\quad\quad\quad\lambda p_A(1,\Tmc)$} (TPlusK)
	(TPlusI) edge[bend left=20, above] node {$\gamma$} (T)
	(TPlusK) edge[bend left=20, below] node {$\gamma$} (T)
	;
	\end{tikzpicture}
	\caption{CTMC with a RRJ ($m>1$).}
	\label{fig:RRJCTMCm}
	\end{minipage}
\end{figure} 

\subsection{Supervised hypothesis testing of Markov chain models}
\label{sec:NPTestMC}

Since the interaction between stations is modeled by Markov chains, we pose the jamming detection problem as a binary hypothesis testing problem, namely, to identify which of two Markov chains is more likely to have produced the sequence of observed states. This section introduces the general problem of binary hypothesis testing of Markov chains, while \secref{modelSetup} presents the setup of the hypothesis testing problem, and the theoretical distribution of the test statistic.  Then, \secref{distributionTestStatistics} develops an analytical upper bound on the variance of the test statistic. 

\subsubsection{The hypothesis test statistic}
\label{sec:modelSetup}

To facilitate analyzing the hypothesis testing problem, we transform the CTMC models proposed in \secref{mathematicalModel} to discrete-time Markov chains (DTMCs) through {\em uniformization}. Choose an uniformization parameter $u$ obeying $u\geq \underset{i,j,\bsf}{\max}\, |q^{\bsf}_{i,j}|$. Following the standard procedure, 
we obtain a DTMC with transition matrix $\Pbf^{\bsf} = \Ibf_{d+1}+\frac{\Qbf^{\bsf}}{u}$, where $\Ibf_{d+1}$ denotes the $d+1$-dimensional identity matrix. We suppose the network monitor at the AP  collects the transmission patterns of all stations in the network over a time interval of length $\frac{W}{u}$, consisting of $W$ intervals each of length $1/u$.  Each observation is sampled from the continuous-time stochastic process generated by a Markov chain in each time interval. To test whether or not the SUT $1$ is a RRJ, we collect a sequence of observations of all stations' transmission on-off processes, and form the observation sequence of $\ybf_W\equiv \{Y(t)\}_{t=1}^{W+1}$, called a \emph{sample path} of the DTMC, generated by a network containing a RRJ or containing only compliant stations. Under the assumption that the transition probability matrices $\Pbf^{\zerosf}$ of compliant DTMC and $\Pbf^{\onesf}$ of RRJ DTMC are known, a binary hypothesis testing problem can be formed:
\begin{equation}
\Hsf_{\zerosf}: \Pbf=\Pbf^{\zerosf}, ~ \Hsf_{\onesf}: \Pbf=\Pbf^{\onesf}.
\end{equation}
This formulation requires supervised training, {\em i.e.}, the detector uses both normal and attack samples for training. The monitor may use observations of the transmission patterns of the network both with and without RRJs to estimate the transition probabilities.  

We now derive the likelihood of $\ybf_W$ under $\Hsf_{\zerosf}$ and $\Hsf_{\onesf}$. Define $i)$ the state transition counts as the number of $\ybf_W$'s transitions from state $i$ to $j$, denoted $N_{i,j}\equiv \sum_{t=1}^{W}\mathbf{1}_{\{y_t=i,y_{t+1}=j\}}$, where $y_t$ indicates the $t$-th element of vector $\ybf_W$, and $ii)$ the state occupancy counts, denoted $N_i\equiv \sum_{j=0}^d N_{i,j} =  \sum_{t=1}^{W}\mathbf{1}_{\{y_t=i\}} $. The log-likelihood of $\ybf_W$ under hypothesis $\Hsf_{\bsf}$ is: $\ln f(\ybf_W|\Hsf_{\bsf}) = \ln\pi^{\bsf}_{y_1}\prod_{t=1}^W p^{\bsf}_{y_t,y_{t+1}} = \ln\pi_{y_1}^{\bsf}+\sum_{i=0}^d\sum_{j=0}^d N_{i,j}\ln p^{\bsf}_{i,j}$,
with $\pi_{y_1}^{\bsf}$ denotes the stationary distribution of state $y_1$ of the DTMC under $\Hsf_{\bsf}$, and $p_{i,j}^{\bsf}$ is the transition probability from state $i$ to state $j$. Thus, the log-likelihood ratio (LLR) between $\Hsf_{\onesf}$ and $\Hsf_{\zerosf}$ is \cite{Lev2008}:
$\ln \frac{f(\ybf_{W}|\Hsf_{\onesf})}{f(\ybf_W|\Hsf_{\zerosf})}$, and the LLR test is:
\begin{equation}
\label{eq:NPtest}
\begin{aligned}
\ln\frac{\pi^{\onesf}_{y_1}}{\pi^{\zerosf}_{y_1}} + \sum_{i=0}^d\sum_{j=0}^d N_{i,j}\ln\left(\frac{p_{i,j}^{\onesf}}{p_{i,j}^{\zerosf}}\right)\,\, & \mathop{\gtreqless}_{\Hsf_{\zerosf}}^{\Hsf_{\onesf}}\,\, \xi(W),\\
\sum_{i=0}^d\sum_{j=0}^d \frac{N_{i,j}}{W}\ln\left(\frac{p_{i,j}^{\onesf}}{p_{i,j}^{\zerosf}}\right) \,\,  & \mathop{\gtreqless}_{\Hsf_{\zerosf}}^{\Hsf_{\onesf}}\,\,  \xi'\equiv \frac{\xi(W) -\ln\frac{\pi^{\onesf}_{y_1}}{\pi^{\zerosf}_{y_1}}}{W}.
\end{aligned}
\end{equation}
Note that the threshold $\xi(W)$ varies with the observation window length $W$ to balance the tradeoff between the false alarm rate and the missed detection rate.  The most natural choice is $\xi(W) = \xi_0 W$ for which $\xi' \approx \xi_0$ for $W$ large.  That is,  as evident from \eqref{NPtest}, the initial distribution has an influence on the detection threshold that decreases to $0$ in $W$ as $1/W$, and as such has little impact on the test outcome for large $W$.  The test statistics of the above LLR test is
\begin{equation}
\label{eq:defZ}
Z \equiv \sum_{i=0}^d\sum_{j=0}^d l_{i,j}\frac{N_{i,j}}{W},
\end{equation} 
with parameters $l_{i,j}\equiv \ln\left(\frac{p_{i,j}^{\onesf}}{p_{i,j}^{\zerosf}}\right)$. To derive the distribution of $Z$ under $\Hsf_{\bsf}$ requires the distributions of $N_{i,j}$. As is well-known \cite{Bar1951}, the $N_{i,j}$ have an asymptotic (in $W$) normal distribution. As a linear combination of $N_{i,j}$, the test statistic $Z$ is therefore also asymptotically normal. We henceforth use the superscript $\bsf\in\{\zerosf,\onesf\}$ to indicate a certain random variable under hypothesis $\Hsf_{\bsf}$. The expectation of $N^{\bsf}_{i,j}$ is $\Ebb[N_{i,j}^{\bsf}]=W\pi_i^{\bsf}p_{i,j}^{\bsf}$, 
and the variance of $N^{\bsf}_{i,j}$ may be written as (c.f. \cite{Bil1961})
\begin{equation}
\label{eq:VarNij}
{\rm Var}[N_{i,j}^{\bsf}] = W(\pi^{\bsf}_ip^{\bsf}_{i,j} - {\pi_i^{\bsf}}^2 {p_{i,j}^{\bsf}}^2) + 2 \pi_i^{\bsf} {p_{i,j}^{\bsf}}^2\sum_{t'=1 }^{W-1}(W-t')\epsilon_{j,i}^{\bsf}{}^{(t'-1)}.
\end{equation}
Here, $[\Pbf]_{i,j}$ denotes the $(i,j)$-th entry of a matrix $\Pbf$, ${\Pbf^{\bsf}}^{(t'-1)}$ is the $(t'-1)$-step transition matrix under $\Hsf_{\bsf}$, and $\epsilon_{j,i}^{\bsf}{}^{(t'-1)}\equiv [{\Pbf^{\bsf}}^{(t'-1)}]_{j,i}-\pi^{\bsf}_i$ \cite{Gan1955}. As $W\uparrow \infty$, transition counts $N_{i,j}$, $N_{i',j'}$ are asymptotically independent for $i\neq i'$, and thus we assume that ${\rm Cov}[N_{i,j},N_{i',j'}]=0$ for $i\neq i'$.  We have derived the covariance of $N_{i,j}$ between $N_{i',j'}$ when $i=i',j\neq j'$ as
\begin{equation}
\label{eq:CovNij}
 {\rm Cov}[N_{i,j}^{\bsf},N_{i',j'}^{\bsf}] =  \pi_i^{\bsf} p_{i,j}^{\bsf}p_{i,j'}^{\bsf}\left(-W\pi_i^{\bsf} + \sum_{ t'=1}^{W-1}(W-t') \left({\epsilon_{j,i}^{\bsf}}^{(t'-1)}+{\epsilon_{j',i}^{\bsf}}^{(t'-1)}\right)\right),
\end{equation}
but the derivation of \eqref{CovNij} is omitted due to space. Based on the above properties of $N_{i,j}$, the distribution of $Z^{\bsf}$ under $\Hsf_{\bsf}$ is asymptotically normal with mean $\mu_Z^{\bsf}\equiv \sum_{i,j} l_{i,j}\pi_i^{\bsf}p_{i,j}^{\bsf}$.  Defining $Z_i^{\bsf}=\sum_{j=0}^d l_{i,j}\frac{N_{i,j}^{\bsf}}{W}$, the variance of $Z^{\bsf}$ is derived as follows:
\begin{equation}
\label{eq:VarZ}
\begin{aligned}
{\rm Var}[Z^{b}]  & = \sum_{i=0}^d{\rm Var}[Z_i^{\bsf}] + \sum_{i\neq j}{\rm Cov}[Z_i^{\bsf},Z_j^{\bsf}]\\
 \overset{(a)}{\approx} & \sum_{i=0}^d{\rm Var}[Z_i^{\bsf}] = \sum_{i} \left(\sum_{j} l_{i,j}^2\frac{{\rm Var}[N_{i,j}^{\bsf}]}{W^2}+2\sum_{ j< j'} l_{i,j}l_{i,j'} \frac{{\rm Cov}[N_{i,j}^{\bsf},N_{i,j'}^{\bsf}]}{W^2}\right)\\
= &  \sum_{i, j}l_{i,j}^2\pi_i^{\bsf}p_{i,j}^{\bsf}  \left(\frac{1-\pi_i^{\bsf}p_{i,j}^{\bsf}}{W}+2 p_{i,j}^{\bsf} \sum_{t'=1 }^{W-1}\frac{W-t'}{W^2}\epsilon_{j,i}^{\bsf}{}^{(t'-1)}\right)+\\
& 2\sum_{i, j< j'}  l_{i,j}l_{i,j'}\pi_i^{\bsf} p_{i,j}^{\bsf}p_{i,j'}^{\bsf}\left(-\frac{\pi_i^{\bsf}}{W}+\sum_{t'=1}^{W-1} \frac{W-t'}{W^2}\left({\epsilon_{j,i}^{\bsf}}^{(t'-1)}+{\epsilon_{j',i}^{\bsf}}^{(t'-1)}\right) \right),
\end{aligned}
\end{equation}
Note the approximation $(a)$ holds for large $W$ by the asymptotic independence of $N_{i,j}$'s.

\subsubsection{Analytical upper bound on the variance of the test statistic}
\label{sec:distributionTestStatistics}

The previous section showed that test statistic $Z$ is asymptotically normal as $W\uparrow \infty$, and we used this to derive the asymptotic mean and variance of $Z$ under the null and alternative hypotheses. However, the expression of ${\rm Var}[Z^{\bsf}]$ are unwieldy and difficult to compute as they includes a summation over the $t$-step transition matrices, which are $t$-th powers of the (one-step) transition matrix. In this section we develop a simpler, and more easily computable, upper bound on the variance of $Z$. As $Z$ is a linear combination of the transition counts, it is useful to first derive an upper bound on the variances and covariances of $\frac{N_{i,j}}{W}$. Xue {\em et al.\ } proposed spectral and graphical bounds for the error covariance measure of the classic steady-state distribution estimator of DTMCs in \cite{XueRoy2011}. However, they did not analyze the variances and covariances of the transition counts and the variance of the test statistic, which is the focus of our work. 
\begin{lemma}
\label{lem:VarCovNijUpperbound}
For an ergodic DTMC with simple eigenvalues, the variance of $\frac{N_{i,j}^{\bsf}}{W}$ has upper bound:
\begin{equation}
\label{eq:VarNijUpperbound}
{\rm Var}\left[\frac{N_{i,j}^{\bsf}}{W}\right] \leq  \pi_i^{\bsf}p_{i,j}^{\bsf}\left(\frac{(1-\pi_i^{\bsf}p_{i,j}^{\bsf})}{W} + 2p_{i,j}^{\bsf}c_{j,i}\frac{2+W|1-\lambda_1^{\bsf}|}{W^2|1-\lambda_1^{\bsf}|^2}\right).
\end{equation}
Here, $c_{j,i}\equiv \sum_{r=1}^d |u_{jr}v_{ri}|$ is a constant that depends upon the right eigenvector matrix $\Ubf\in \Rbb^{(d+1)\times (d+1)}$ of $\Pbf^{\bsf}$, $u_{jr}$ denotes the $(j,r)$-th entry of the $\Ubf$, $v_{ri}$ denotes the $(r,i)$-th entry of $\Ubf^{-1}$, and $\lambda_1^{\bsf}$ represents the largest non-unit eigenvalue of the transition matrix $\Pbf^{\bsf}$, in the sense that $|1-\lambda_1^{\bsf}|\leq |1-\lambda_r^{\bsf}|$ where $\lambda_r^{\bsf}$ denotes all the other non-unit eigenvalues of $\Pbf^{\bsf}$. Similarly, the covariance of $N_{i,j}$ and $N_{i,j'}$ ($j\neq j'$) has upper bound:
\begin{equation}
{\rm Cov}\left[\frac{N_{i,j}^{\bsf}}{W},\frac{N_{i,j'}^{\bsf}}{W}\right] \leq \pi_i^{\bsf}p_{i,j}^{\bsf}p_{i,j'}^{\bsf}\left(-\frac{ \pi_i^{\bsf}}{W} + \left(c_{j,i}+c_{j',i}\right)\frac{2+W|1-\lambda_1^{\bsf}|}{ W^2|1-{\lambda_1^{\bsf}}|^2} \right).
\end{equation}
\end{lemma}

The proof is shown in \secref{proofVarCovNijUpperbound}.

\begin{prop}
\label{prop:VarZUpperbound}
When $\Pbf^{\zerosf}$, $\Pbf^{\onesf}$ are ergodic and have simple eigenvalues, ${\rm Var}[Z^{\bsf}]$ has upper bound:
\begin{equation}
\label{eq:VarZUpperbound}
\begin{aligned}
& {\rm Var}[Z^{\bsf}] \leq \sum_{i, j}l_{i,j}^2\pi_i^{\bsf}p_{i,j}^{\bsf}\left(\frac{1-\pi_i^{\bsf}p_{i,j}^{\bsf}}{W} + 2p_{i,j}^{\bsf}c_{j,i}\frac{2+W|1-\lambda_1^{\bsf}|}{W^2|1-\lambda_1^{\bsf}|^2}\right) \\
& +2 \sum_{ i, j< j'} \max{\{0,l_{i,j}l_{i,j'}\}}\pi_i^{\bsf} p_{i,j}^{\bsf}p_{i,j'}^{\bsf} \left(\frac{-\pi_i^{\bsf} }{W} + \left(c_{j,i}+c_{j',i}\right)\frac{2+W|1-\lambda_1^{\bsf}|}{W^2|1-\lambda_1^{\bsf}|^2} \right).
\end{aligned}
\end{equation}
\end{prop}

\begin{proof}
The result follows immediately by applying \lemref{VarCovNijUpperbound} to \eqref{VarZ}.
\end{proof}

We can see that $\lim_{W\rightarrow \infty} {\rm Var}[Z^{\bsf}] = 0$ with convergence rate $O(1/W)$.  The upper bound in \eqref{VarZUpperbound} is a complicated function of the eigenvectors and the largest non-unit eigenvalue of $\Pbf^{\bsf}$.

\subsection{Strategies for the RRJ under full observability}
\label{sec:rrjStrategy}

While the previous subsection proposes a jamming detector based upon a Markov model of the carrier sensing mechanism of a network, this section studies how an intelligent RRJ should choose the best operating point. As an attacker, the RRJ naturally has two objectives: $i)$ maximizing its jamming efficiency, and $ii)$ minimizing the probability of being detected.  With this objective in mind,  \secref{performanceMetrics} proposes two performance metrics of an RRJ, \secref{LDP} develops a large deviations principle (LDP) approximation for the detection probability, and \secref{jointlyChoosePRPJ} formulates the RRJ's optimization problem based on the LDP. 

\subsubsection{Performance metrics of RRJ}
\label{sec:performanceMetrics}

Two performance metrics are proposed with regard to the RRJ's two objectives: $i)$ the fraction of collision time caused by the SUT; $ii)$ the detector's error probability. Define a {\em ``collision state''} as a state with multiple transmitters. In this case, the probability of a SUT being in a collision state is: $r^{\bsf}=\sum_{i\in\Smc_3} \pi_{i}^{\bsf} = \boldsymbol{\pi}^{\bsf}\tbf$, where $\Smc_3\equiv \{\Tmc \in \Smc, 1\in \Tmc, |\Tmc|>1 \}$ denotes the set of collision states of the Markov chain involving the SUT, and vector $\tbf\in \Rbb^{2^m\times 1}$ has components $t_i=1$ if $i \in \Smc_3$ and $t_i=0$ otherwise. The {\em jamming efficiency metric} is defined as: 
\begin{equation}
\label{eq:etaDef}
\eta \equiv r^{\onesf}/r^{\zerosf}.
\end{equation}
The second performance metric is defined either as the {\em missed-detection rate (MDR)} or as the {\em equal error rate (EER)}, depending upon the context. Given a threshold $\xi'$,  the {\em MDR} is $p_{\RN{2}}(\xi')\equiv \Pbb[Z^{\onesf}\leq \xi']$, and the FAR is $p_{\RN{1}}(\xi') \equiv \Pbb[Z^{\zerosf}>\xi']$. The EER is defined as $\frac{p_{\RN{1}}(\xi^*) + p_{\RN{2}}(\xi^*)}{2}$ where $\xi^*$ is the specific detection threshold at which the MDR is equal to the FAR\footnote{Or has the minimum distance to the FAR if equality cannot be obtained; this scenario explains the EER definition.}

Define the column vector $\pbf_{\rm rrj}=[p_R, p_J] \in [0,1]\times [0,1]$ as holding the two design parameters $(p_R,p_J)$. We assume the RRJ is aware of the design of the jamming detector, and as such its objective is to {\em maximize} the detection error (which may be captured by the MDR), under the constraint that its jamming efficiency, $\eta$, is above certain efficiency threshold, denoted $\tau_{\eta}$. To emphasis the fact that the transition probability matrix of the DTMC with a RRJ and its corresponding stationary distribution are parameterized by the jamming probabilities, hereafter we denote the transition matrix and its stationary distribution as $\Pbf^{\onesf}(\pbf_{\rm rrj})$ and $\boldsymbol{\pi}^{\onesf}(\pbf_{\rm rrj})$ respectively.  Thus, the RRJ has the following optimization problem:
\begin{equation}
\label{eq:RRJobjective}
\underset{\pbf_{\rm rrj}\in [0,1]\times [0,1]}{\mbox{maximize}} \quad
\null p_{\RN{2}}(\xi'_{\alpha}) \quad \text{subject to} \quad \boldsymbol{\pi}^{\onesf}(\pbf_{\rm rrj})\tbf \geq \tau_{\eta} \boldsymbol{\pi}^{\zerosf}\tbf.
\end{equation}
Here, $i)$ $f_{Z|\Hsf_{\onesf}}(z)$ denotes the Gaussian PDF $\Nca(\mu_Z^{\onesf}, {\rm Var}[Z^{\onesf}])$, $ii)$ threshold $\xi'_{\alpha}\equiv \Phi^{-1}(1-\alpha)\sqrt{{\rm Var}[Z^{\zerosf}]}+\mu_Z^{\zerosf}$ (with $\Phi^{-1}$ the inverse CDF of the standard Gaussian distribution), and thus $iii)$ the detection objective may be written as $p_{\RN{2}}(\xi_{\alpha}')=\int_{-\infty}^{\xi'_{\alpha}}f_{Z|\Hsf_{\onesf}}(z)\drm z$. 

In computing $p_{\RN{2}}(\xi'_{\alpha})$ in the objective function, we have shown that deriving the exact distribution of $Z^{\bsf}$ involves calculating powers of the transition matrix. To simplify computation, one option is to use the upper bound derived in \secref{distributionTestStatistics} to approximate ${\rm Var}[Z^{\bsf}]$. In this case, we first analyze the effect of applying the upper bound of ${\rm Var}[Z^{\bsf}]$ to the MDR computation. Define the upper bound of ${\rm Var}[Z^{\bsf}]$ shown in \eqref{VarZUpperbound} as $\hat{\sigma}_Z^{\bsf}{}^2$ for $\bsf=\zerosf, \onesf$, and define a random variable $\hat{Z}^{\bsf}\sim \Nca(\mu_Z^{\bsf},{\hat{\sigma}_Z^{\bsf}}{}^2)$ with PDF $f_{\hat{Z}^{\bsf}}(z)$. Then we have $\hat{\xi'}_{\alpha}\equiv  \Phi^{-1}(1-\alpha)\hat{\sigma}_Z^{\zerosf}+\mu_Z^{\zerosf} \geq \xi'_{\alpha}$, and thus, $\int_{-\infty}^{\xi'_{\alpha}}f_{Z^{\onesf}}(z)\drm z \leq \int_{-\infty}^{\hat{\xi'}_{\alpha}}f_{Z^{\onesf}}(z)\drm z \leq \int_{-\infty}^{\hat{\xi'}_{\alpha}}f_{\hat{Z}^{\onesf}}(z)\drm z$. Hence the usage of $\hat{\sigma}_Z^{\bsf}$ will significantly inflate the MDR in the objective function \eqref{RRJobjective} if the upper bound in \eqref{VarZUpperbound} is not tight enough. Furthermore, the (approximated) MDR in the objective, $ p_{\RN{2}}(\xi'_{\alpha})$ ({\em i.e.}, $\int_{-\infty}^{\hat{\xi'}_{\alpha}}f_{\hat{Z}^{\onesf}}(z)\drm z$), cannot be expressed as a explicit function of the RRJ design parameters. Thus, it is infeasible to directly solve the optimization problem in \eqref{RRJobjective} and obtain an explicit optimal jamming strategy $(p_R^*, p_J^*)$. As such, we consider an alternative objective function based upon large deviations theory.

\subsubsection{Large deviations principle (LDP) for the asymptotic missed detection rate (MDR)}
\label{sec:LDP}

The G\"{a}rtner-Ellis theorem generalizes Cram\'{e}r's theorem, which gives the decay rate of the probability that the empirical mean of independent and identically distributed (i.i.d.) random variables deviates from the expectation, to the non-i.i.d.\ case \cite{DemZei2010}. It has been shown in \cite{Lev2008} that the moment generating function (MGF) of $Z^{\zerosf}$ converges to $\Lambda(t)\equiv \ln \lambda_{\max}(t)$ where $\lambda_{\max}(t)$ is the eigenvalue with the largest magnitude of the matrix $\Pbf(t)$, whose $(i,j)$-th entry is $p_{i,j}^{\onesf}{}^{(t)} p_{i,j}^{\zerosf}{}^{(1-t)}$. Similarly, the MGF of $Z^{\onesf}$ is $\Lambda(t+1)=\ln\lambda_{\max}(t+1)$ \cite{Lev2008}. The Fenchel-Legendre transform of $\Lambda(t)$ is: $\Lambda^*(\xi)\equiv \sup_{t\in\Rbb} (\xi t-\Lambda(t)) = \sup_{t\in\Rbb} (\xi t- \ln\lambda_{\max}(t))$. According to the G\"{a}rtner-Ellis theorem, the decaying rate of the FAR and MDR for detection threshold $\xi'\in (\mu_Z^{\zerosf},\mu_Z^{\onesf})$ are \cite{Lev2008}: $ \lim\limits_{W\rightarrow \infty}\frac{\ln p_{\RN{1}}(\xi')}{W}=-\Lambda^*(\xi')$, and $\lim\limits_{W\rightarrow \infty}\frac{1}{W}\ln p_{\RN{2}}(\xi')=\xi'-\Lambda^*(\xi')$.  Therefore, the FAR and MDR can be represented by $p_{\RN{1}}(\xi') = c_{\RN{1}}(W) e^{-W\Lambda^*(\xi')}$, and $p_{\RN{2}}(\xi') = c_{\RN{2}}(W) e^{-W(\Lambda^*(\xi')-\xi')}$, with constants $\lim\limits_{W\rightarrow \infty}\ln c_{\RN{1}}(W)=0$, and similarly for $c_{\RN{2}}(W)$. Furthermore, the following lemma derives the convergence rate of the infimum MDR corresponding to a bounded FAR.

\begin{lemma}[Stein's Lemma \cite{DemZei2010}]
\label{lem:steinLemma}
Let $0<\alpha<1$ denote a FAR threshold, and let $p_{\RN{2}}^{\alpha}\equiv \inf_{\xi\in \{\xi':p_{\RN{1}}(\xi')<\alpha\}} p_{\RN{2}}(\xi)$ denote the corresponding minimum MDR.  Then the asymptotic (in $W$) minimum MDR corresponding to FAR $\alpha$ is given by the LDP $\lim\limits_{W\rightarrow \infty}\frac{1}{W} \ln p_{\RN{2}}^{\alpha}  = -I(\Pbf^{\zerosf},\Pbf^{\onesf}(\pbf_{\rm rrj}))$, with rate function $I(\Pbf^{\zerosf},\Pbf^{\onesf}(\pbf_{\rm rrj}))\equiv \sum_{i,j} \pi_i^{\zerosf} p_{i,j}^{\zerosf}\ln\frac{p_{i,j}^{\zerosf}}{p_{i,j}^{\onesf}}$.
\end{lemma}

According to \lemref{steinLemma}, the infimum MDR for a test that guarantees the $p_{\RN{1}}(\xi')\leq \alpha$ can be expressed as $p_{\RN{2}}^{\alpha}= c_{\RN{2}}(W)e^{-I(\Pbf^{\zerosf},\Pbf^{\onesf}(\pbf_{\rm rrj}))W}$. Since $\lim_{W\rightarrow \infty}\ln c_{\RN{2}}(W)=0$, we can ignore $c_{\RN{2}}(W)$ by assuming that $c_{\RN{2}}(W)=1$. The following Lemma and proposition show that the function $I(\Pbf^{\zerosf},\Pbf^{\onesf}(p_R, p_J))$ is a convex function of the RRJ design parameters $p_R, p_J$.

\begin{lemma}
\label{lem:convergenceRateConvex}
The rate function $I(\Pbf^{\zerosf},\Pbf^{\onesf}(\pbf_{\rm rrj}))$ for the asymptotic (in $W$) minimum MDR of an RRJ is a convex function of $p_R$ and $p_J$.
\end{lemma}

The proof is in \secref{proofConvergenceRateConvex}.  Observe $I(\Pbf^{\zerosf},\Pbf^{\onesf}(\pbf_{\rm rrj}))$ obtains its unique minimum of zero at $p_R=1, p_J=0$, where the RRJ behaves exactly the same as a CS-compliant station.

\subsubsection{Asymptotic (in $W$) optimal choice of RRJ design parameters $(p_R,p_J)$}
\label{sec:jointlyChoosePRPJ}

We consider the RRJ's objective of finding the (asymptotic in $W$) optimal $(p_R^*,p_J^*)$, {\rm i.e.,} to minimize the rate function (which maximizes the MDP) subject to the efficiency constraint:
\begin{equation}
\begin{aligned}
\underset{\pbf_{\rm rrj}\in [0,1]\times [0,1]}{\mbox{minimize}} \quad 
&\null I(\Pbf^{\zerosf}, \Pbf^{\onesf}(\pbf_{\rm rrj}))\quad \mbox{subject to}\quad \boldsymbol{\pi}^{\onesf}(\pbf_{\rm rrj})\tbf \geq \tau_{\eta} \boldsymbol{\pi}^{\zerosf}\tbf.
\end{aligned}
\end{equation}
There are two potential issues with the constraint in the minimization problem above: $i)$ to the best of our knowledge, the leading eigenvector of the transition matrix, $\boldsymbol{\pi}^{\onesf}(\pbf_{\rm rrj})$ generally does not have an explicit expression when the dimension of $\Pbf^{\onesf}(\pbf_{\rm rrj})$ is high; and $ii)$ the efficiency constraint produces a non-convex feasible set. To address both these issues, we approximate $\boldsymbol{\pi}^{\onesf}(\pbf_{\rm rrj})$ via Taylor series expansion.  This expansion both approximates the stationary distribution as a function of $\pbf_{\rm rrj}$, and also converts the non-convex set to a convex feasible set (at least for the first-order expansion).  The order-$k$ (for $k \in \{1,2\}$) Taylor series expansion of  $\boldsymbol{\pi}^{\onesf}$ around the point $\hat{\pbf}_{\rm rrj} = [\hat{p_R}, \hat{p_J}] \in (0,1]^2$ is 
\begin{equation}
\boldsymbol{\pi}_{\rm ts}^{\onesf}(\pbf_{\rm rrj}, k) 
= \boldsymbol{\pi}^{\onesf}(\hat{\pbf}_{\rm rrj}) + (\pbf_{\rm rrj} - \hat{\pbf}_{\rm rrj})\left(\Jbf_{\boldsymbol{\pi}^{\onesf}}(\hat{\pbf}_{\rm rrj}) + \frac{\mathbf{1}_{\{k=2\}}}{2} \Hbf_{\boldsymbol{\pi}^{\onesf}}(\hat{\pbf}_{\rm rrj}) [\Ibf_{d+1}\otimes (\pbf_{\rm rrj} - \hat{\pbf}_{\rm rrj})^{\intercal}]\right)
\end{equation}
with Jacobian $\Jbf_{\boldsymbol{\pi}^{\onesf}} (\pbf_{\rm rrj}) \in \Rbb^{2}$, and Hessian $\Hbf_{\boldsymbol{\pi}^{\onesf}}(\pbf_{\rm rrj})\in \Rbb^{2\times 2}$ ($\otimes$ denotes the Kronecker product). 

The $k$-th order partial derivative of $\boldsymbol{\pi}^{\onesf}(\pbf_{\rm rrj})$ w.r.t.\ $p_R$ is $\frac{\partial^k\boldsymbol{\pi}^{\onesf}(\pbf_{\rm rrj})}{\partial p_R^k} = k!(-1)^k \boldsymbol{\pi}^{\onesf}(\pbf_{\rm rrj})\left(\frac{\partial\Qbf^{\onesf}(\pbf_{\rm rrj})}{\partial p_R}\Gbf(\pbf_{\rm rrj})\right)^k$, in which $\Gbf(\pbf_{\rm rrj})$ denotes the group inverse matrix of $\Qbf^{\onesf}(\pbf_{\rm rrj})$ and $\frac{\partial\Qbf^{\onesf}(\pbf_{\rm rrj})}{\partial p_R}$ denotes the first order partial derivative w.r.t.\ $p_R$ \cite{DhoChe2013}.  The same result applies for the $k$-th order derivative w.r.t.\ $p_J$. The mixed second-order derivative is \cite{DhoChe2013}: 
\begin{equation}
\frac{\partial^2\boldsymbol{\pi}^{\onesf}(\pbf_{\rm rrj})}{\partial p_R \partial p_J} = \boldsymbol{\pi}^{\onesf}(\pbf_{\rm rrj})\left(\frac{\partial\Qbf^{\onesf}(\pbf_{\rm rrj})}{\partial p_R}\Gbf(\pbf_{\rm rrj})\frac{\partial\Qbf^{\onesf}(\pbf_{\rm rrj})}{\partial p_J}\Gbf(\pbf_{\rm rrj})+\frac{\partial\Qbf^{\onesf}(\pbf_{\rm rrj})}{\partial p_J}\Gbf(\pbf_{\rm rrj})\frac{\partial\Qbf^{\onesf}(\pbf_{\rm rrj})}{\partial p_R}\Gbf(\pbf_{\rm rrj})\right).
\end{equation}
For $k\in \{1,2\}$, the approximate optimization problem becomes:
\begin{equation}
\label{eq:approxOptimization}
\begin{aligned}
\underset{\pbf_{\rm rrj}\in [0,1]\times [0,1]}{\mbox{minimize}} \quad 
&\null I(\Pbf^{\zerosf}, \Pbf^{\onesf}(\pbf_{\rm rrj}))\quad \mbox{subject to}\quad \boldsymbol{\pi}_{\rm ts}^{\onesf}(\pbf_{\rm rrj}, k)\tbf \geq \tau_{\eta} \boldsymbol{\pi}^{\zerosf}\tbf.
\end{aligned}
\end{equation}

\begin{theorem}
\label{thm:rrjOptimization}
The optimization problem \eqref{approxOptimization} is a convex optimization problem when $k=1$.
\end{theorem}

\begin{proof}
According to \lemref{convergenceRateConvex}, the objective function in \eqref{approxOptimization} is a convex function of $\pbf_{\rm rrj}$. Also when $k=1$, $\boldsymbol{\pi}^{\onesf}_{\rm ts}$ is the first-order Taylor series expansion to $\boldsymbol{\pi}^{\onesf}$ at point $\hat{\pbf}_{\rm rrj}$. Therefore the inequality constraint is linear and also convex in $\pbf_{\rm rrj}$. 
\end{proof}

\subsection{Semi-supervised hypothesis testing of Markov chain models}
\label{sec:semiSupvised}

The hypothesis testing problem proposed in \secref{modelSetup} will not be applicable if the RRJ is not available for training. An alternative test that requires only the knowledge of transition matrix corresponding to the network without an RRJ, termed semi-supervised testing, is of natural interest.  Intuitively, it is desired that the detector report an anomaly upon noticing a deviation in the transmission pattern from $\Pbf^{\zerosf}$.  To derive such a detector, we cast the problem as a goodness-of-fit test:
$\Hsf_{\zerosf}: \Pbf = \Pbf^{\zerosf}$ vs.\ $\Hsf_{\onesf}: \Pbf \neq \Pbf^{\zerosf}$.
The log-likelihood of observing a sample path $\ybf_{W}$ with transition counts $\{N_{i,j}\}$ is $\ln f(\ybf_{W}) = \ln \pi^{\zerosf}_{y_1} + \sum_{i,j} N_{i,j}\ln p_{i,j}^{\zerosf}$.
The goodness-of-fit test between $\Hsf_{\zerosf}$ and $\Hsf_{\onesf}$ is
\begin{equation}
\label{eq:semiSupervisedTest}
\sum_{i,j} \frac{N_{i,j}}{W}\ln p_{i,j}^{\zerosf} \,\,\,\mathop{\gtreqless}_{\Hsf_{\onesf}}^{\Hsf_{\zerosf}}\,\,\, \frac{\xi(W) - \ln \pi^{\zerosf}_{y_1}}{W},
\end{equation}
and the test statistic is $Z= \sum_{i,j} \frac{N_{i,j}}{W}\ln p_{i,j}^{\zerosf}$.  The only difference between the semi-supervised test statistic \eqref{semiSupervisedTest} and the supervised test statistic \eqref{defZ} is that the LLR coefficients $l_{i,j}$'s in the latter are replaced by the coefficients $\ln p_{i,j}^{\zerosf}$'s in the former. Therefore, following the same approach employed in \secref{NPTestMC}, it is straightforward to derive the mean and variance of the test statistic under $\Hsf_{\bsf}$.  In particular, $\Ebb[Z^{\bsf}] = \sum_{i,j} \pi_{i}^{\bsf}p_{i,j}^{\bsf}\ln p_{i,j}^{\zerosf}$, and
\begin{equation}
\begin{aligned}
\mbox{Var}[Z^{\bsf}] & =  \sum_{i, j}(\ln p_{i,j}^{\zerosf})^2 \pi_i^{\bsf}p_{i,j}^{\bsf}  \left(\frac{1-\pi_i^{\bsf}p_{i,j}^{\bsf}}{W}+2 p_{i,j}^{\bsf} \sum_{t'=1 }^{W}\frac{W-t'}{W^2}\epsilon_{j,i}^{\bsf}{}^{(t'-1)}\right)+\\
& 2\sum_{i, j< j'}  (\ln p^{\zerosf}_{i,j}\ln p^{\zerosf}_{i,j'})\pi_i^{\bsf} p_{i,j}^{\bsf}p_{i,j'}^{\bsf}\left(-\frac{\pi_i^{\bsf}}{W}+\frac{W-t'}{W^2}\left({\epsilon_{j,i}^{\bsf}}^{(t'-1)}+{\epsilon_{j',i}^{\bsf}}^{(t'-1)}\right)  \right).
\end{aligned}
\end{equation}
In addition, the analytical upper bound on the test statistic proposed in \propref{VarZUpperbound} is applicable to the semi-supervised approach's test statistic if we replace the $l_{i,j}$'s in \eqref{VarZUpperbound} by $\ln p_{i,j}^{\zerosf}$'s. 

We next develop new models in \secref{limitedModel}, but in \secref{Experiments} we will $i)$ evaluate the accuracy of the test statistic variance upper bound from \secref{distributionTestStatistics}, $ii)$ evaluate the accuracy of the Taylor series approximation to the intelligent RRJ strategy optimization problem from \secref{jointlyChoosePRPJ}, and $iii)$ compare the detection accuracy of the supervised and semi-supervised hypothesis tests from \secref{semiSupvised}.

\section{Markov models for limited observability}
\label{sec:limitedModel}

The ``full-observability'' assumed in \secref{fullModel}, {\rm i.e.,} knowledge of which stations are transmitting at any time, requires that AP and RRJ have the strong capability to differentiate all $m$ of the signals emitted from the stations.  Such an assumption is not feasible in practice.  As such, this section makes the weaker assumption that the AP and RRJ only have {\em limited observability}.  They don't know {\em which} stations are currently transmitting, but instead they only know partial information, such as, $i)$ the number of stations currently transmitting (the ``intermediate'' model, \secref{intermediateModel}), or $ii)$ whether or not any stations are currently transmitting (the ``simplified'' model, \secref{simpleModel}).  

\subsection{Intermediate model: knowledge of the number of transmitters}
\label{sec:intermediateModel}

To compute the test statistic $Z^{\bsf}$ in \secref{fullModel}, the AP must know network parameters $\lambda,\gamma$ and the station location vector $\xbf$ in order to compute the log-likelihood coefficients $l_{i,j}$.  It must also know the transmission status (active or idle) of each station, at every instant in time, in order to compute the transition counts $N_{i,j}$'s of the full sample path.  This is a strong assumption about the capability of the AP which may not hold in practice.   This subsection considers an ``intermediate'' model in which the AP only holds two pieces of information: $i)$ the number of stations currently transmitting over the channel, $C \equiv |\Tmc|$ (recall $\Tmc$ denotes the set of currently active stations); and $ii)$ whether or not the SUT (station $1$) is transmitting, $X \equiv \mathbf{1}_{\{1\in \Tmc\}}$. The state in this model is represented by these two numbers: $(C,X)$, which we assume the AP is able to observe.  The state is a stochastic process with a (largely reduced) state space of size $2m$. 

\subsubsection{State aggregation}
\label{sec:stateAggregation} 

The intermediate model partitions the states of the full observability model, $\Smc$, into the several disjoint subsets. With $\hat{Y} \equiv (C, X)$ denoting a state in the intermediate model, the subset of states in $\Smc$ corresponding to $\hat{Y}$ is denoted $\Smc_{\hat{Y}}$, and consists of all states $\Tmc$ from $\Smc$ in which $|\Tmc| = C$ and $\mathbf{1}_{\{1 \in \Tmc\}} = X$.  This partition of the  state space is called the {\em intermediate partition}. \figref{compliantCTMCIntermediate} shows the state space and transitions of the intermediate model with $m=3$ stations.  

The resulting aggregated stochastic process is denoted $\{\hat{Y}(t): t\geq 0\}$, taking value in the {\em aggregated state space} $\hat{\Smc}=\{(0,0), (1,0), (1,1),...,(m,1)\}$.  Although the process $Y(t)$ is Markov, the aggregated process $\hat{Y}(t)$ may {\em not} be Markov.  A common term for (Markov chain) state aggregation is {\em lumping}, and it is common to call a Markov chain {\em lumpable} w.r.t.\ a certain state partition if the resulting lumped stochastic process is Markov \cite{BalYeo1993}.  We show below that the full model is {\em not strongly lumpable} \cite{BalYeo1993, ZhaWeb2009} w.r.t.\ the proposed intermediate state partition. 

\begin{definition}[\cite{BalYeo1993}]
\label{def:strongLump}
A CTMC is called \textbf{strongly lumpable} w.r.t.\ a partition $\{\Smc_k, k=1,...,K\}$ of its state space $\Smc$ if and only if the rate matrix $\Qbf$ of the CTMC satisfies that for each pair of partitions $\Smc_k$, $\Smc_{k'}$ and for any pair of states $i,j\in \Smc_{k}$, $\sum_{i'\in\Smc_{k'}} q_{i,i'} = \sum_{i'\in \Smc_{k'}} q_{j,i'}$.
\end{definition}

\textbf{Observation:} The full-observability model CTMC is not strongly lumpable w.r.t.\ the state partition induced by the intermediate model, as shown by considering a network with $m=3$ compliant stations. The full model state space is $\Smc=\{\phi, \{1\}, \{2\}, \{3\}, \{1,2\}, \{1,3\}, \{2,3\}, \{1, 2, 3\}\}$.  Consider two sets of states in the full model that map to state $(1,0)$ and $(2,1)$ in the intermediate model: $\Smc_{1,0}=\left\{\{2\}, \{3\}\right\}$, $\Smc_{2,1}=\{\{1,2\}, \{1,3\}\}$. For $i = \{2\}$, $\sum_{i'\in \Smc_{2,1}} q_{i,i'}= q_{\{2\},\{1,2\}} + q_{\{2\},\{1,3\}} = \lambda p_I(1,\{2\})$, while for state $j=\{3\}$, we have $\sum_{i'\in \Smc_{2,1}} q_{j,i'}= q_{\{3\},\{1,2\}} + q_{\{3\},\{1,3\}} = \lambda p_I(1,\{3\})$. By \eqref{idleProb}, $p_I(1,\{2\})\neq p_I(1,\{3\})$ when station 1 has different distances to stations 2 and 3.  Thus, the condition in \defref{strongLump} is violated, primarily due to the dependence of the idle to active transition rates on the set of current active stations. 

\begin{figure}[!htb]
	\vspace{+1mm}
	\begin{minipage}[b]{0.5\textwidth}
		\centering
		\begin{tikzpicture}[font=\sffamily,scale=0.30]
		\tikzset{node style/.style={state, 
				minimum width=1cm,
				line width=0.3mm,
				fill=gray!20!white}}
		\node[node style] at (0, 0)      (00)     {$0,0$};
		\node[node style] at (5.5, 3)      (10) 	  {$1,0$};
		\node[node style] at (5.5, -3)     (11)     {$1,1$};
		\node[node style] at (11, 3)     (20)     {$2,0$};
		\node[node style] at (11, -3)    (21)     {$2,1$};
		\node[node style] at (16.5, 0)    (31)     {$3,1$};
		\draw[every loop,
		auto=right,
		line width=0.3mm,
		>=latex,
		draw=black,
		fill=black]
		(00)     edge[bend right=20]     node {} (10)
		(10)     edge[bend right=20, auto=left] node {} (00)
		(10)     edge[bend right=20]     node {} (21)
		(21)     edge[bend right=20, auto=left] node {} (10)
		(10) edge[bend right=20] node {} (20)
		(20) edge[bend right=20, auto=left] node {} (10)
		(20) edge[bend right=20] node {} (31)
		(31) edge[bend right=20, auto=left] node {} (20)
		(00) edge[bend right=20] node {} (11)
		(11) edge[bend right=20, auto=left] node {} (00)
		(11) edge[bend right=20] node {} (21)
		(21) edge[bend right=20, auto=left] node {} (11)
		(21) edge[bend right=20] node {} (31)
		(31) edge[bend right=20, auto=left] node {} (21)
		;
		\end{tikzpicture}
		\caption{CTMC of intermediate model ($m=3$).}\label{fig:compliantCTMCIntermediate}
	\end{minipage}
	\quad
	\begin{minipage}[b]{0.48\textwidth}
		\centering
		\begin{tikzpicture}[font=\sffamily,scale=0.28]
		\tikzset{node style/.style={state, 
				minimum width=1cm,
				line width=0.3mm,
				fill=gray!20!white}}
		\node[node style] at (0, 0)     (00)     { $0,0$};
		\node[node style] at (0, 6.5)     (10) 	 { $1,0$};
		\node[node style] at (6.5, 0)     (01)     { $0,1$};
		\node[node style] at (6.5, 6.5)     (11)     { $1,1$};
		\draw[every loop,
		auto=right,
		line width=0.3mm,
		>=latex,
		draw=black,
		fill=black]
		(00)     edge[bend left=20, auto=left, sloped, anchor=center, above]     node {} (10)
		(10)     edge[bend left=20, auto=left, sloped, anchor=center, above] 	   node {} (00)
		(10)     edge[bend left=20, auto=left, sloped, anchor=center, above]     node {} (11)
		(11)     edge[bend left=20, auto=left, sloped, anchor=center, above]     node {} (10)
		(11) 	 edge[bend left=20, auto=left, sloped, anchor=center, above] 	   node {} (01)
		(01) 	 edge[bend left=20, auto=left, sloped, anchor=center, above]     node {} (11)
		(00) 	 edge[bend left=20, auto=left, sloped, anchor=center, above]     node {} (01)
		(01) 	 edge[bend left=20, auto=left, sloped, anchor=center, above]     node {} (00)
		;
		\end{tikzpicture}
		\caption{CTMC of simplified model.}\label{fig:compliantCTMCSimplified}
	\end{minipage}
\end{figure}

As the Markov chain under the full observability model is not strongly lumpable with respect to the state partition induced by the intermediate model, the stochastic process for the intermediate model is not Markov.  As our purpose is {\em not} to precisely model the aggregated stochastic process, bu rather to detect the presence or absence of an RRJ, we {\em approximate} the intermediate model's stochastic process as a Markov process, and derive an optimal detector under this assumption.  We demonstrate the accuracy of this approximation in \secref{Experiments}. 

In order to construct the intermediate model Markov chain, we must assign transition rates $\hat{q}_{i, j}$ for $i, j\in \hat{\Smc}$.  For {\em lumpable} chains, the transition rates are obtained by directly summing the transition rates from one aggregate set to another.  As our chain is {\em not} lumpable, however, we employ an alternate method, called {\em ideal aggregate} \cite{Buc1994}.  Under this method, the aggregated transition rate for $\hat{Y}, \hat{Y}' \in \hat{\Smc}$ under hypothesis $\Hsf_{\bsf}$ is:
\begin{equation}
\begin{aligned}
\hat{q}^{\bsf}_{\hat{Y},\hat{Y}'} 
= \frac{\sum_{ \Tmc \in \Smc_{\hat{Y}}} \Pbb(Y=\Tmc |\Hsf_{\bsf}) \sum_{\Tmc' \in \Smc_{\hat{Y}'}} q^{\bsf}_{\Tmc,\Tmc'}}{\sum_{\Tmc \in \Smc_{\hat{Y}}} \Pbb(Y=\Tmc|\Hsf_{\bsf})} 
= \frac{\sum_{ \Tmc \in \Smc_{\hat{Y}}} \pi^{\bsf}_{\Tmc} \sum_{\Tmc' \in \Smc_{\hat{Y}'}} q^{\bsf}_{\Tmc,\Tmc'}}{\sum_{\Tmc\in \Smc_{\hat{Y}}} \pi^{\bsf}_{\Tmc}},
\end{aligned}
\end{equation}
Here, $Y$ denotes the (random) state under the full observability model. One advantage of using the ideal aggregate is its preservation of the jamming efficiency, as shown in \propref{equalEta}. 

\begin{prop}
\label{prop:equalEta}
The full and the ideal aggregated CTMC have the same jamming efficiency $\eta$.
\end{prop}

\begin{proof}
This follows immediately from the fact that the CTMC under ideal aggregation preserves the aggregated stationary distribution of the full model for any state partition \cite{DeMan2008}. 
\end{proof}

\subsection{Simplified model: knowledge of whether or not there are active transmitters}
\label{sec:simpleModel}

In certain practical settings it will be unreasonable to assume that the AP is able to keep track of the number of active transmitters, much less the identities of those transmitters.  As such, this section considers the ``simplified'' model, in which only two bits of information are available to the AP: $i)$ whether or not there are any active transmissions among the CS-compliant stations, and $ii)$ whether or not the SUT is active.  The simplified model aggregates the station status bits $(T_k)_{k=1}^m$ to two bits: $S=\max_{k\in \{2,\ldots,m\}}T_k$ and $X = T_1$. The interaction between the SUT bit $X$ and the CS bit $S$ is captured as a four-state stochastic process $\tilde{Y}(t)$, where $\tilde{Y} = (S,X)$. The simplification is equivalent to partitioning the state space $\Smc$ of the full-observability model into four aggregate sets, $\{\Smc_{0}, \Smc_{1}, \Smc_{2}, \Smc_{3}\}$, where 
$i)$ $\Smc_{0}\equiv \{\emptyset\}$ denotes no station is active; 
$ii)$ $\Smc_{1}\equiv \{\{1\}\}$ holds the state in which the SUT is the only active radio; 
$iii)$ $\Smc_{2}\equiv \{\Tmc\in \Smc: 1 \not\in \Tmc, \Tmc \neq \emptyset\}$ holds states in which the SUT is idle and at least one of the CS-compliant stations is active;
$iv)$ $\Smc_{3} \equiv \{\Tmc\in\Smc: 1 \in \Tmc\}$ holds states in which the SUT and at least one CS-compliant station are active. The proposed state space aggregation 
aligns with our objective to determine whether or not the SUT is CS-compliant or not. The state transition diagram for the simplified model is shown in \figref{compliantCTMCSimplified}.  \tabref{transitionRatesSimplified} gives the transition rates of the Markov chain on the simplified model state space under ideal aggregation (\secref{stateAggregation}), where $p^{\zerosf}(1,\Tmc)=
p_I(1,\Tmc)$ and $p^{\onesf}(1,\Tmc)=p_A(1,\Tmc)$, and the convenience parameters are defined in \tabref{transitionRatesParametersSimplified}. 

\begin{table}[!htb]
	\begin{minipage}[t]{\textwidth}
		\centering
		\begin{tabular}{c|c|c|c|c|c|c|c}
			\hline
			$\tilde{q}^{\bsf}_{00,10}$ & $\tilde{q}^{\bsf}_{10,00}$ &  $\tilde{q}^{\bsf}_{00,01}$ & $\tilde{q}^{\bsf}_{01,00}$ & $\tilde{q}^{\bsf}_{01,11}$ & $\tilde{q}^{\bsf}_{11,01}$ &  $\tilde{q}^{\bsf}_{10,11}$ & $\tilde{q}^{\bsf}_{11,10}$  \\
			\hline
			$(m-1)\lambda$ & $\beta^{\bsf}_{10,00}\lambda$ &  
			$p^{\bsf}(1,\phi)\lambda$ & $\gamma$ & 
			$\beta_{01,11}\lambda$ & $\beta^{\bsf}_{11,01}\lambda$ &
			$\beta^{\bsf}_{10,11}\lambda$ & $\gamma$ \\
			\hline
		\end{tabular}
		\caption{Transition rates for the simplified model.}
		\label{tab:transitionRatesSimplified}
	\end{minipage}
	\begin{minipage}[t]{\textwidth}
		\centering
		\begin{tabular}{c|c|c|c}
			\hline
			$\beta_{01,11}$ & $\beta^{\bsf}_{10,00}$ & $\beta^{\bsf}_{10,11}$ & $\beta^{\bsf}_{11,01}$ \\
			\hline
			$\sum_{i=2}^m p_I(i,\{1\})$ 
			& $\frac{  \sum_{\Tmc \in \Smc_2} \pi^{\bsf}_{\Tmc} \mathbf{1}_{\{|\Tmc|=1\}}}{\sum_{\Tmc \in \Smc_{2}} \pi^{\bsf}_{\Tmc}}$ 
			&  $\frac{ \sum_{\Tmc \in \Smc_{2}} \pi^{\bsf}_{\Tmc} p^{\bsf}(1, \Tmc) }{\sum_{\Tmc \in \Smc_{2}} \pi^{\bsf}_{\Tmc}}$ 
			& $\frac{ \sum_{\Tmc \in \Smc_{3}} \pi^{\bsf}_{\Tmc}\mathbf{1}_{\{|\Tmc|=2\}} }{\sum_{\Tmc \in \Smc_{3}} \pi^{\bsf}_{\Tmc}}$\\
			\hline
		\end{tabular}
		\caption{Convenience parameters in the transition rates for the simplified model.}
		\label{tab:transitionRatesParametersSimplified}
	\end{minipage}
\end{table}

\section{Numerical results}
\label{sec:Experiments}

We now present numerical results; network parameters are shown in \tabref{simulationParameters}\footnote{Parameter values are the defaults in \texttt{ns-3} (https://www.nsnam.org/).}. \figref{systemTopo} shows the network topology considered in the simulation: the topology spacing parameter $R$ controls the number of HT pairs in the network. Due to space limitations, we only show experimental results of $R=40$, which corresponds to stations 1-4 and stations 2-3 being HT pairs, when Rayleigh fading is omitted.

\begin{figure}[!htb]
	\begin{minipage}[b]{0.6\textwidth}
	\centering
	\subfloat[$m=4$]{\includegraphics[scale=0.45]{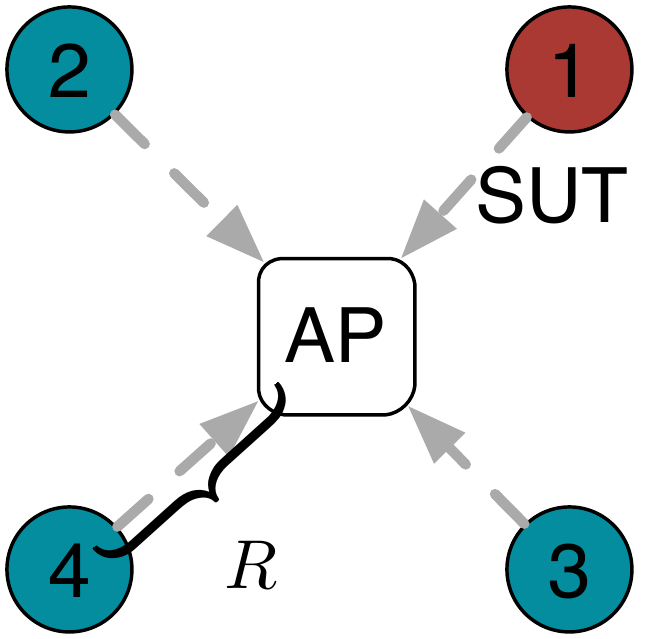}}
	\quad\quad\quad
	\subfloat[$m=6$]{\includegraphics[scale=0.45]{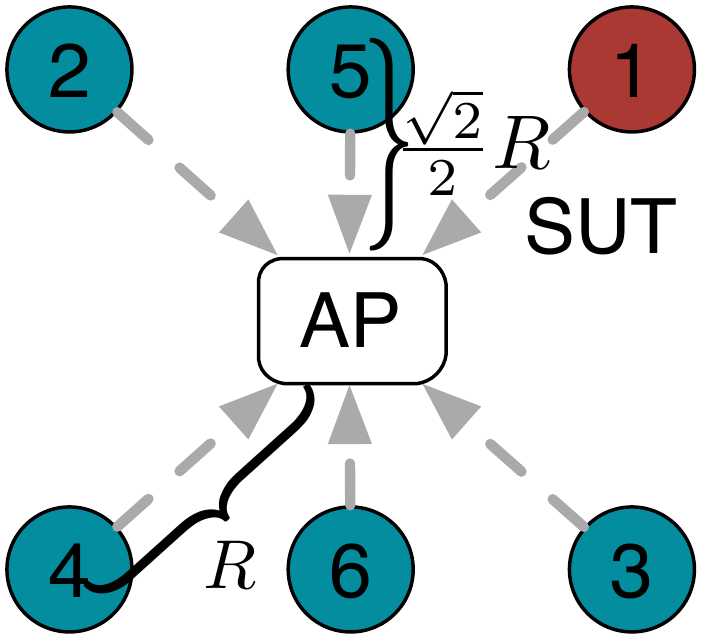}}
	\caption{Wireless network topology.}\label{fig:systemTopo}
	\end{minipage}
	\begin{minipage}[b]{0.35\textwidth}
		\centering
		\begin{tabular}{l|l}
			\hline
			$p_t$ 		& $0.04$ W\\			
			$N_0$ 		& $4.0124*10^{-13}$ W\\
			$\theta$ 	& $2.5119*10^{-12}$ W\\			
			$p_o$ 		& $8.5959*10^{-7}$ W\\		
			$d_o$		& $1$ m\\
			\hline
		\end{tabular}
		\caption{Simulation parameters.}\label{tab:simulationParameters}
	\end{minipage}
	\begin{minipage}[b]{\textwidth}
	\centering
	\subfloat[EER vs. $p_R$ ($p_J=0.01$)\label{fig:eerVsPR}]{\includegraphics[width=0.28\linewidth]{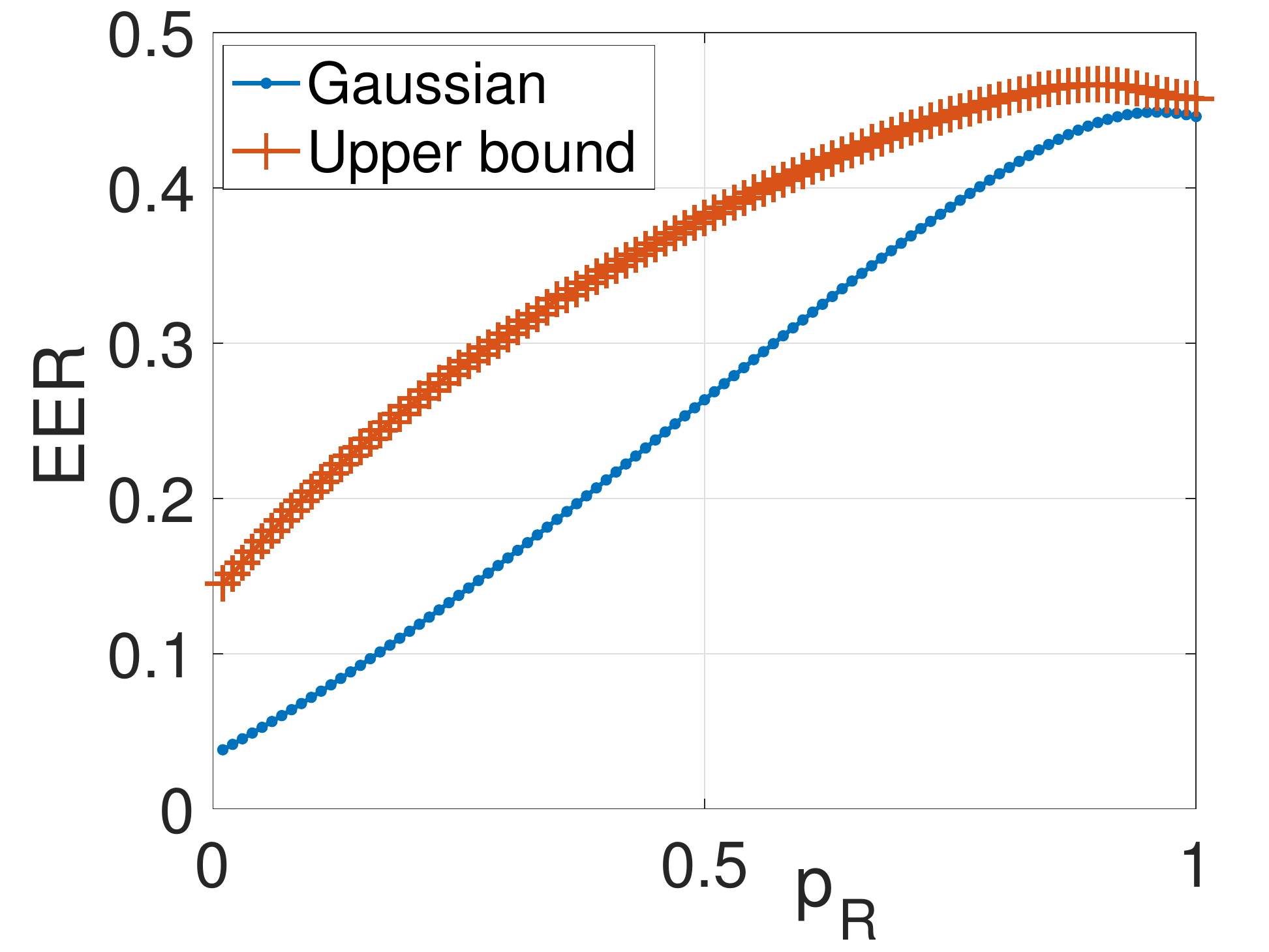}}
	\subfloat[EER vs. $p_J$ ($p_R=0.1$)\label{fig:eerVsPJ}]{\includegraphics[width=0.28\linewidth]{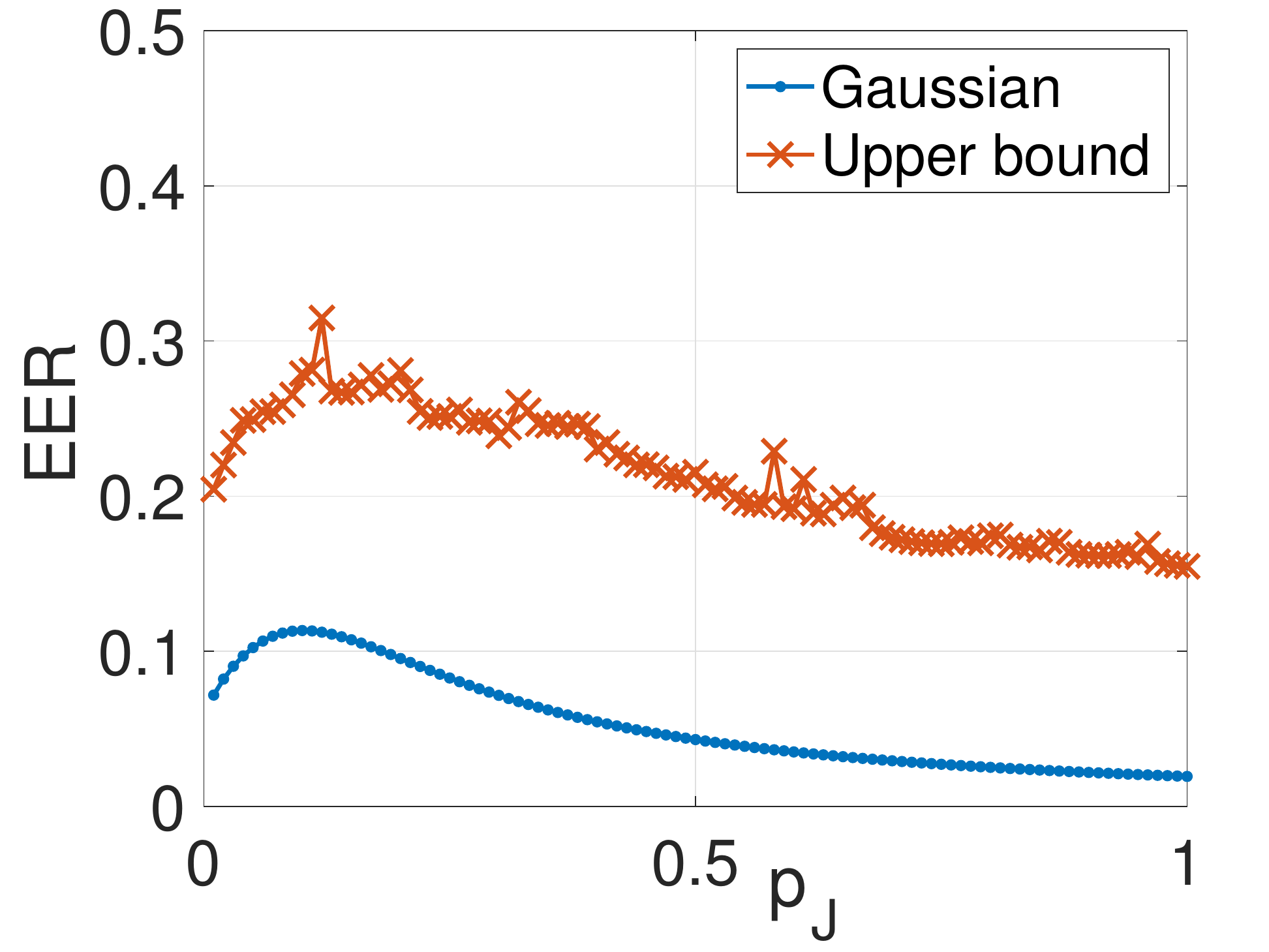}}
	\subfloat[EER vs. $p_J$ ($p_R=1$)\label{fig:eerVsPJ}]{\includegraphics[width=0.28\linewidth]{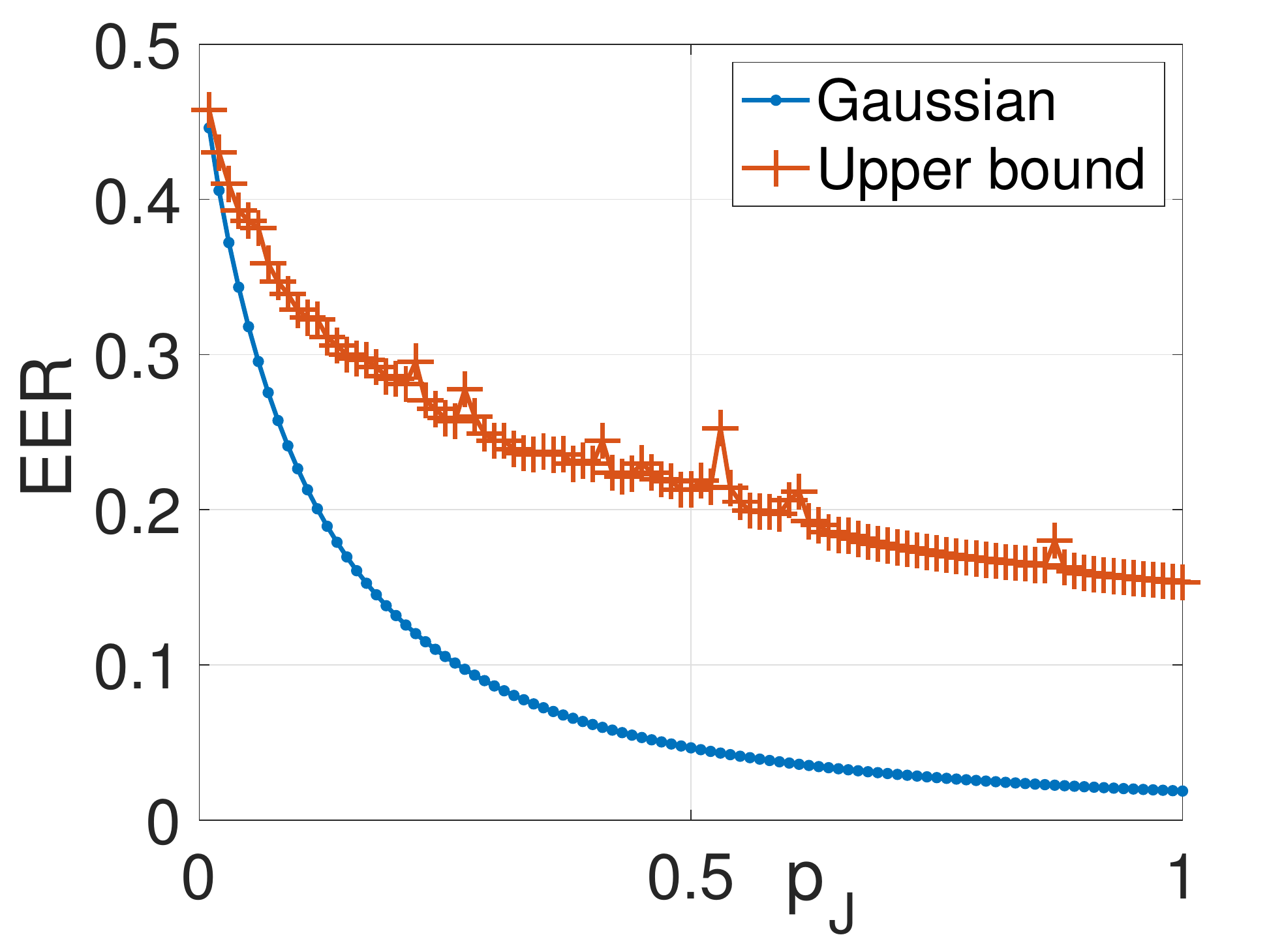}}
	\caption{Relationship between jamming parameters and the EER ($m=6$, $W=1000$).}\label{fig:eerVsPRPJ}	
	\end{minipage}
\end{figure}

\subsection{Impacts of the RRJ parameters on the EER}

This subsection shows how the RRJ parameters $(p_R, p_J)$ affect the detection EER. \figref{eerVsPRPJ} shows the EER computed using the theoretical ${\rm Var}(Z^{\bsf})$ in \eqref{VarZ} (labeled ``Gaussian'') and the EER computed using the upper bound of ${\rm Var}(Z^{\bsf})$ in \propref{VarZUpperbound}. We see the EER is not necessarily a monotone function of either $p_R$ or $p_J$. We also observe the upper bound is somewhat loose when $p_R$ is small (with $p_J$ fixed) and when $p_J$ is large (with $p_R$ fixed).  This looseness provides justification for our development of the LDP in the objective function of \thmref{rrjOptimization}.

\begin{figure}[!htb]
	\begin{minipage}[b]{0.63\textwidth}
		\subfloat[Compare $\eta_{\rm ts}$'s and $\eta$ \label{fig:etaTsComparison}]{\includegraphics[scale=0.25]{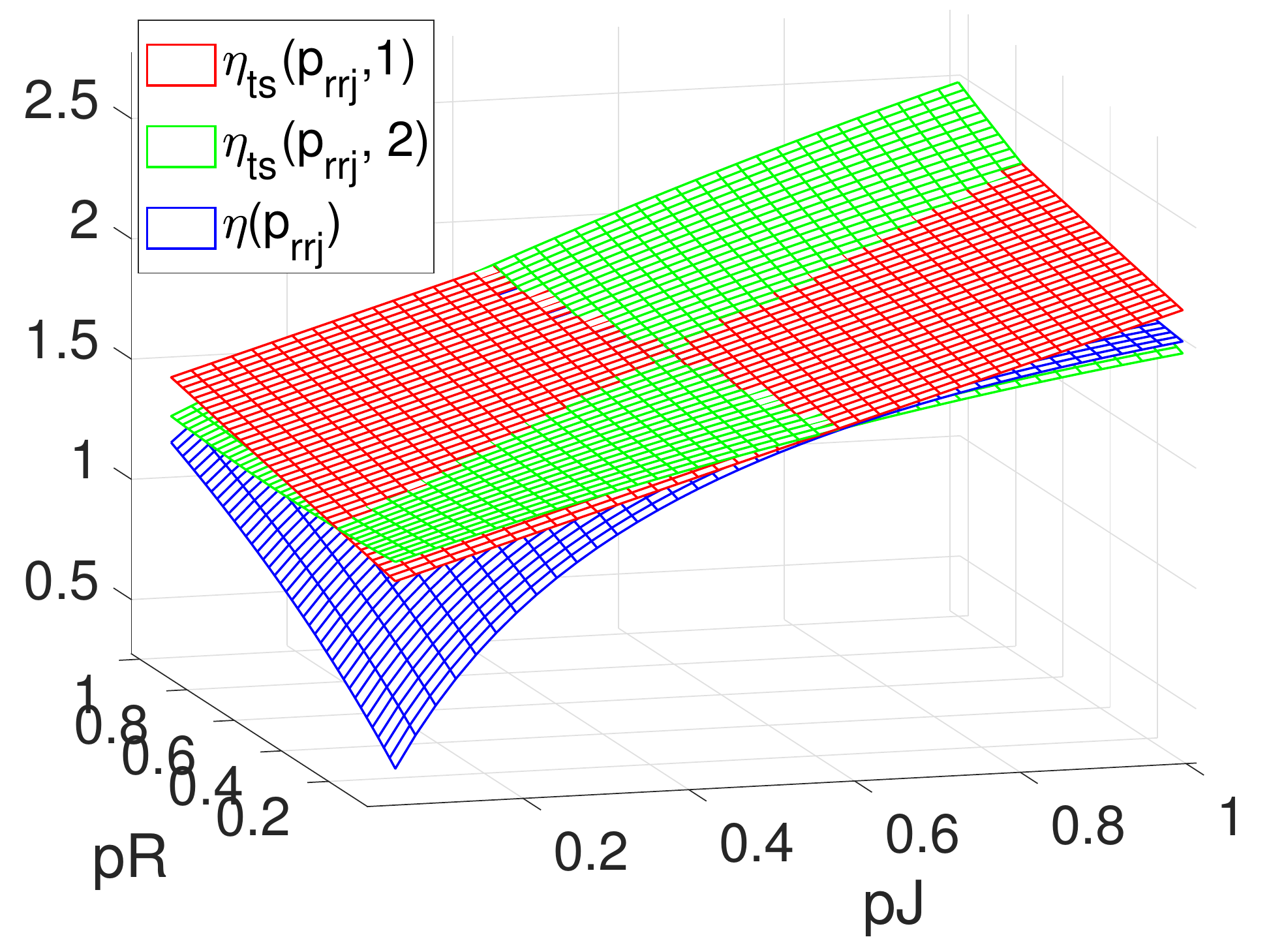}}
		\hfill
		\subfloat[$\frac{|\eta_{\rm ts}(\pbf_{\rm rrj}, 1)-\eta(\pbf_{\rm rrj})|}{\eta(\pbf_{\rm rrj})}$\label{fig:absoluteDiffEtaTs}]{\includegraphics[scale=0.25]{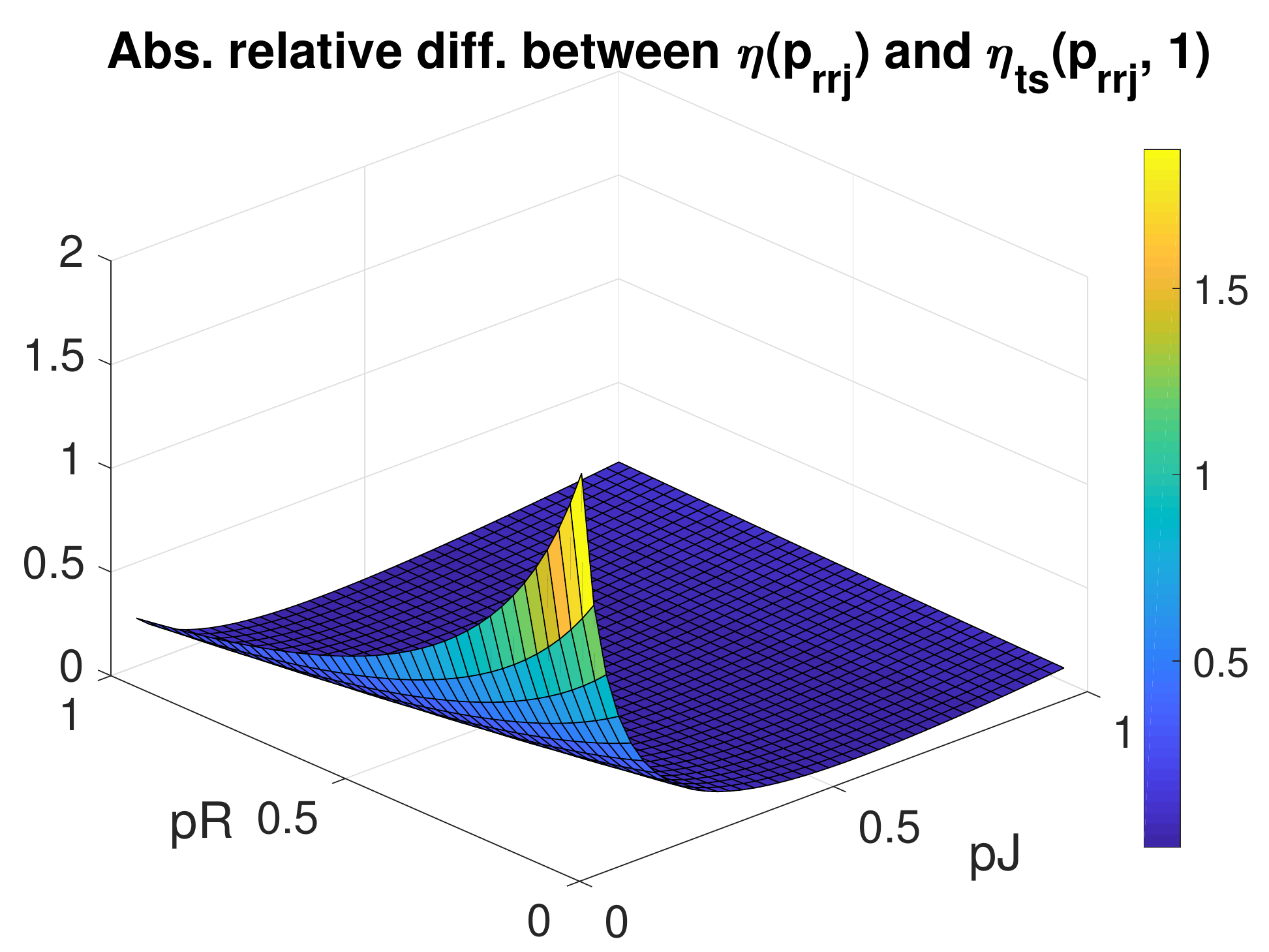}}
		\caption{Compare $\eta_{\rm ts}(\pbf_{\rm rrj}, 1)$, $\eta_{\rm ts}(\pbf_{\rm rrj}, 2)$ with the true $\eta$ ($\hat{\pbf}_{\rm rrj}=[0.5, 0.5]$, $m=6$, $R=40$)}
	\end{minipage}
	\quad\quad
	\begin{minipage}[b]{0.35\textwidth}
		\includegraphics[scale=0.3]{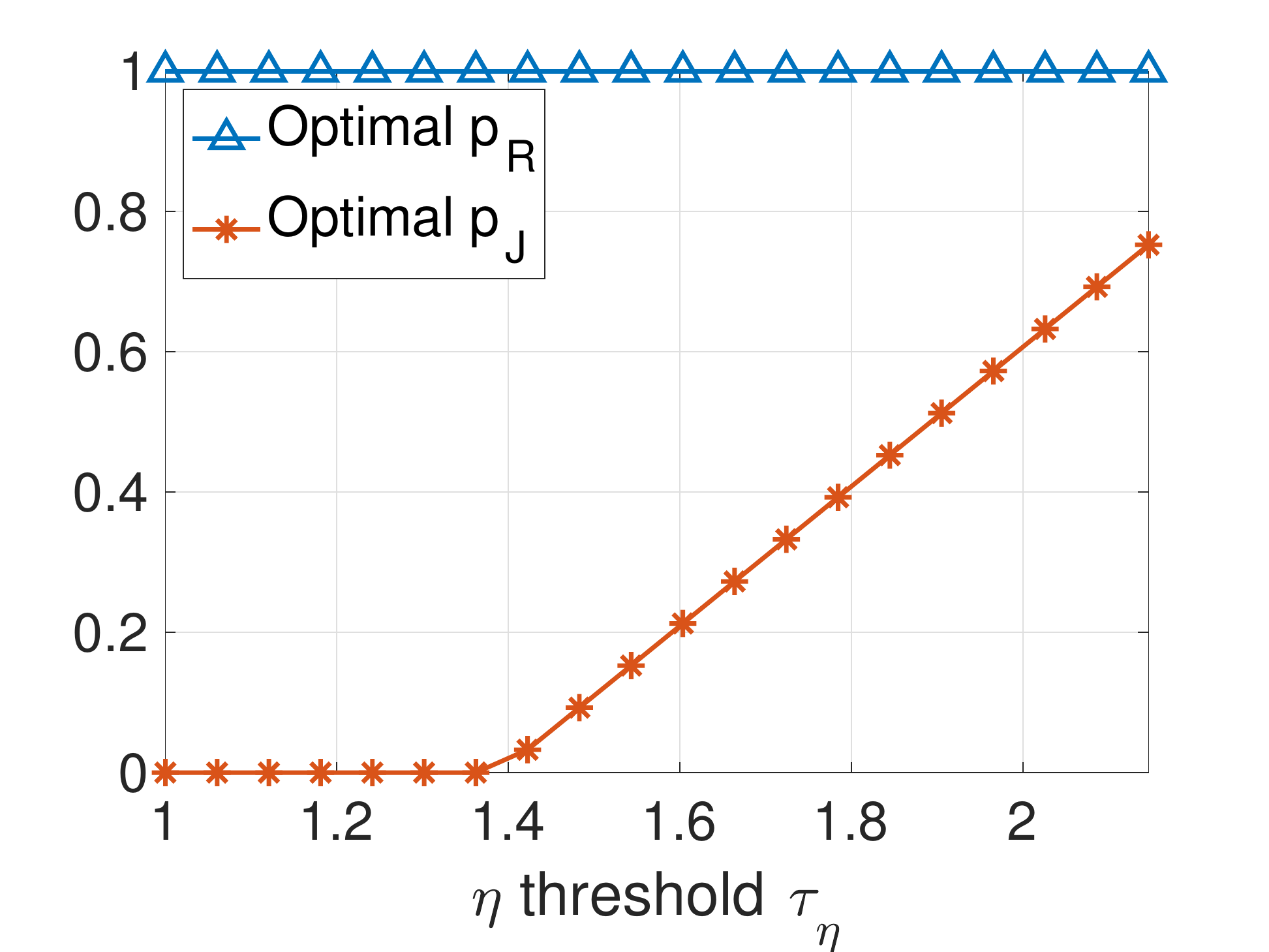}
		\caption{Optimal choices of $p_R$ and $p_J$ vs. $\tau_{\eta}$ ($m=6$, $R=40$).}\label{fig:optimalPRPJ}
	\end{minipage}
	\begin{minipage}[t]{0.32\textwidth}
		\centering
		\includegraphics[scale=0.25]{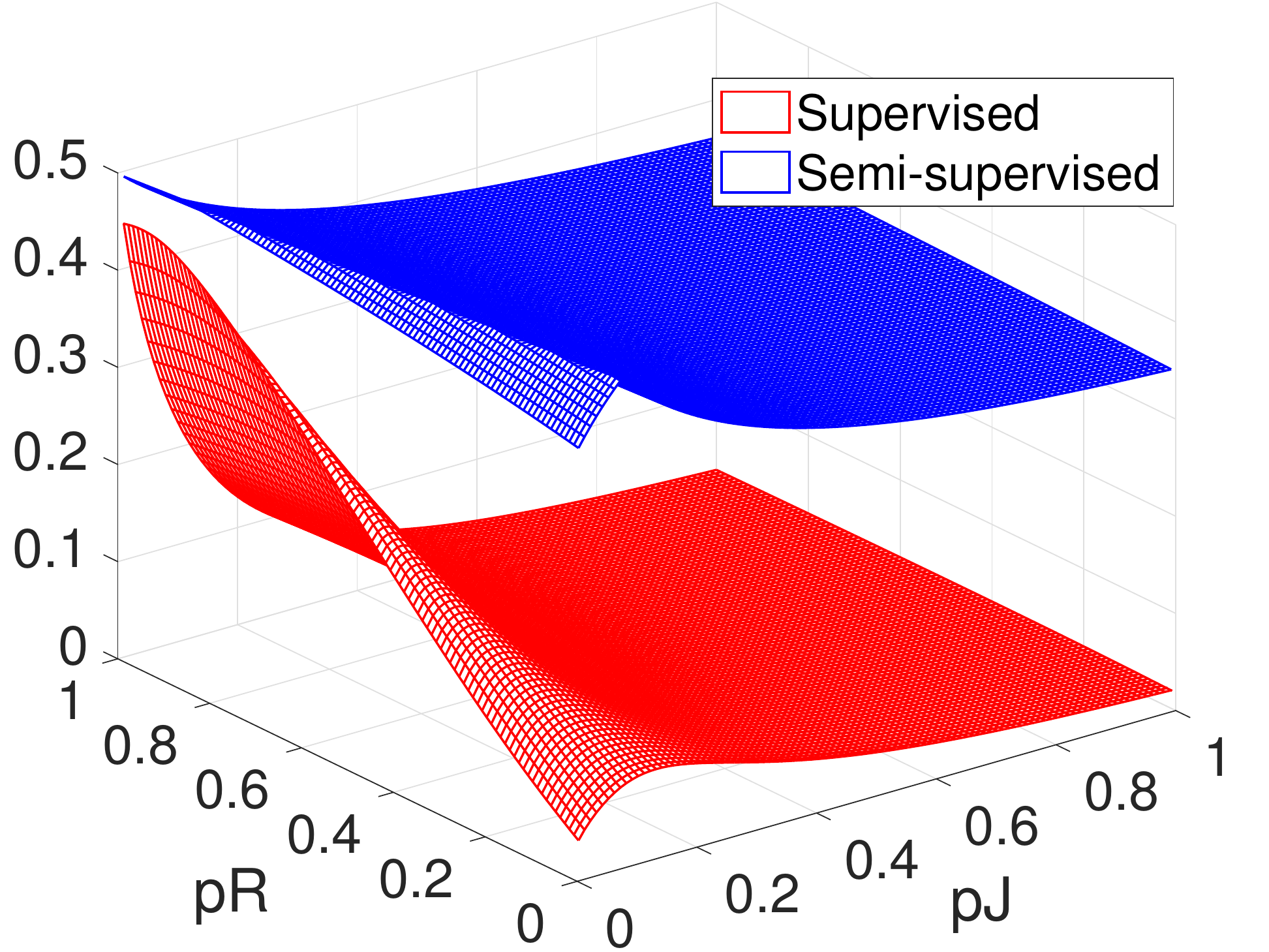}
		\caption{EER vs. $p_R-p_J$ \\($m=6$, $W=1000$).}\label{fig:EER-Compare-supVssemiSup}
	\end{minipage}
	\hfill
	\begin{minipage}[t]{0.32\textwidth}
		\centering
		\includegraphics[scale=0.23]{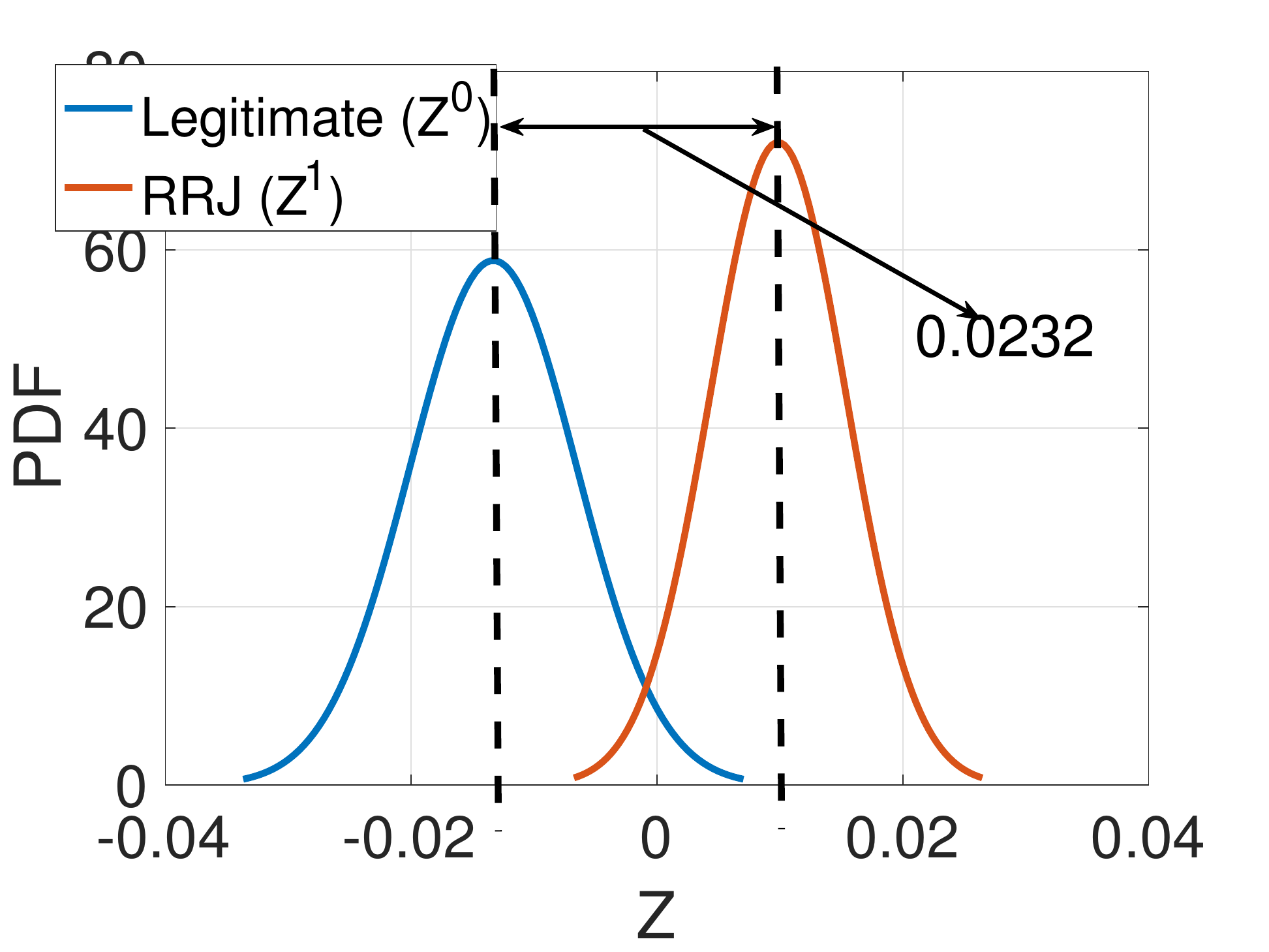}
		\caption{PDFs of $Z^{\zerosf}$, $Z^{\onesf}$ (supervised, $p_R=0.01, p_J=1$).}\label{fig:testStatisticsPDF-sup}
	\end{minipage}
	\hfill
	\begin{minipage}[t]{0.32\textwidth}
		\centering
		\includegraphics[scale=0.23]{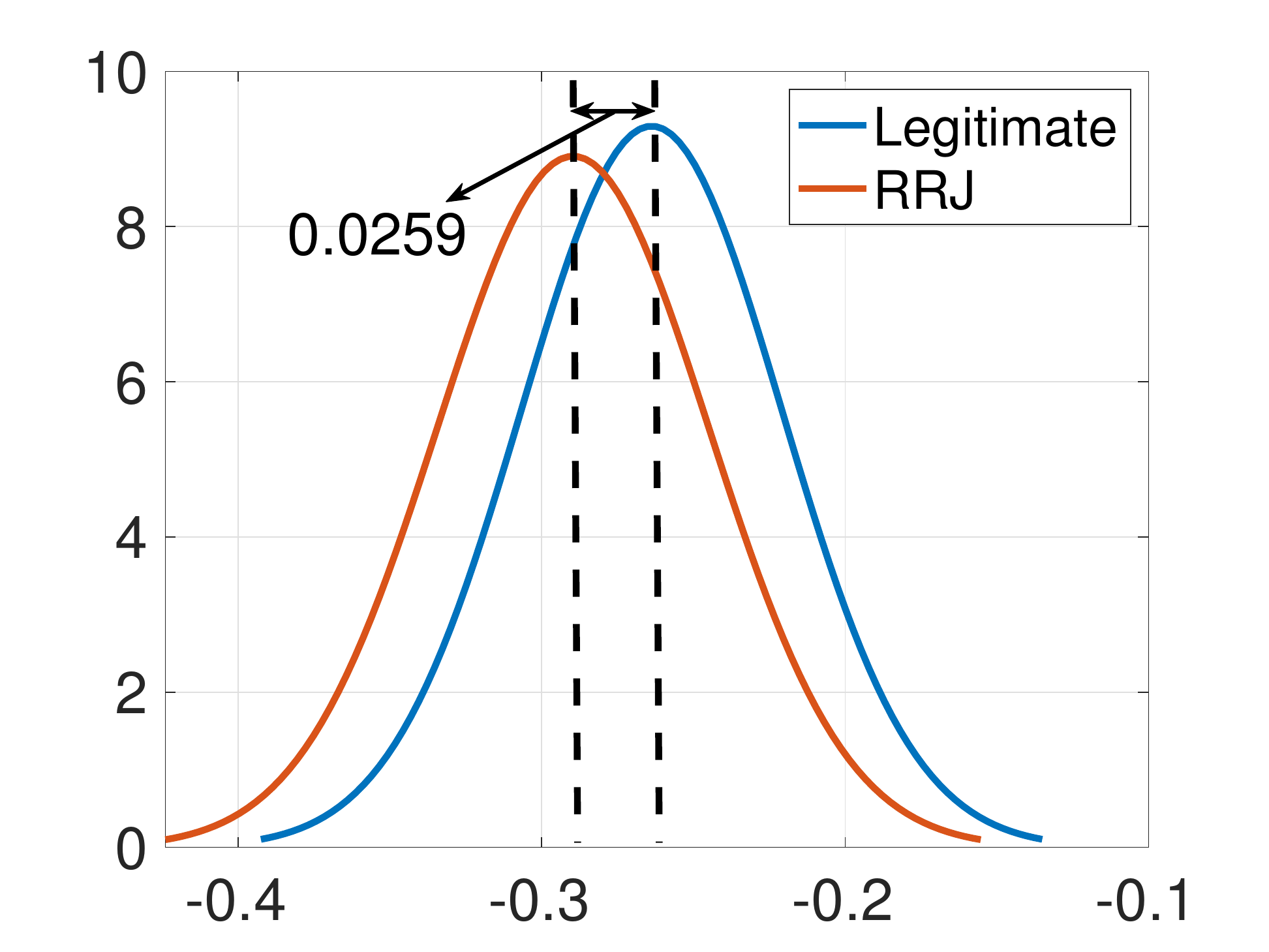}
		\caption{PDFs of $Z^{\zerosf}$, $Z^{\onesf}$ (semi-supervised, $p_R\!\!=\!\!0.01,p_J=1$).}\label{fig:testStatisticsPDF-semiSup}
	\end{minipage}
\end{figure}
\subsection{RRJ strategies under the full observability model}
As proposed in \secref{jointlyChoosePRPJ}, the jamming efficiency constraint may be approximated via Taylor series $\boldsymbol{\pi}^{\onesf}$ to approximately compute the jamming efficiency $\eta(\pbf_{\rm rrj})$. Denote the approximated efficiency obtained using $\boldsymbol{\pi}^{\onesf}_{\rm ts}(\pbf_{\rm rrj}, k)$ as $\eta_{\rm ts}(\pbf_{\rm rrj}, k)$. We first investigates the difference between $\eta$ and $\eta_{\rm ts}(\pbf_{\rm rrj}, 1)$, $\eta_{\rm ts}(\pbf_{\rm rrj}, 2)$ respectively. As shown in \figref{etaTsComparison} (plotted on the $p_R$-$p_J$ plane with grid interval $0.0244$), the difference between the first-order and second-order Taylor series truncation is not very large, and the average relative difference $\frac{|\eta_{\rm ts}(\pbf_{\rm rrj}, 1)-\eta(\pbf_{\rm rrj})|}{\eta(\pbf_{\rm rrj})}$ is 0.0883, and the average relative error $\frac{|\eta_{\rm ts}(\pbf_{\rm rrj}, 2)-\eta(\pbf_{\rm rrj})|}{\eta(\pbf_{\rm rrj})}$ is 0.0828. \figref{absoluteDiffEtaTs} shows the absolute relative difference between $\eta_{\rm ts}(\pbf_{\rm rrj}, 1)$ and $\eta(\pbf_{\rm rrj})$, and we can see that the discrepancy is small when $p_R$ and $p_J$ are large. \figref{optimalPRPJ} presents the optimal $p_R$ and $p_J$ obtained by solving the optimization problem \eqref{approxOptimization}, and shows that the optimal solution is always obtained at $p_R=1$ and the choice of $p_J$ depends on the jamming efficiency threshold. However, the efficiency constraint in \eqref{approxOptimization} uses $\boldsymbol{\pi}^{\onesf}_{\rm ts}(\pbf_{\rm rrj}, 1)$ which overestimates $\eta$, especially when $p_R$ and $p_J$ are small by \figref{etaTsComparison}. Therefore, the optimization yields solutions that have a slightly lower jamming efficiency than $\tau_{\eta}$.

\subsection{Comparison the supervised and the semi-supervised models}

This experiment compares the detection performance of the supervised detector proposed in \secref{NPTestMC} and the semi-supervised detector proposed in \secref{semiSupvised}. We can see from \figref{EER-Compare-supVssemiSup} that the semi-supervised anomaly detector has a much higher EER than the supervised detector. Comparison between \figref{testStatisticsPDF-sup} and \figref{testStatisticsPDF-semiSup} shows the performance degradation is due to the fact that the semi-supervised test statistics have a much larger variance than the supervised ones.

\begin{figure}[!htb]
	\begin{minipage}[b]{0.48\textwidth}
		\centering
		\includegraphics[scale=0.28]{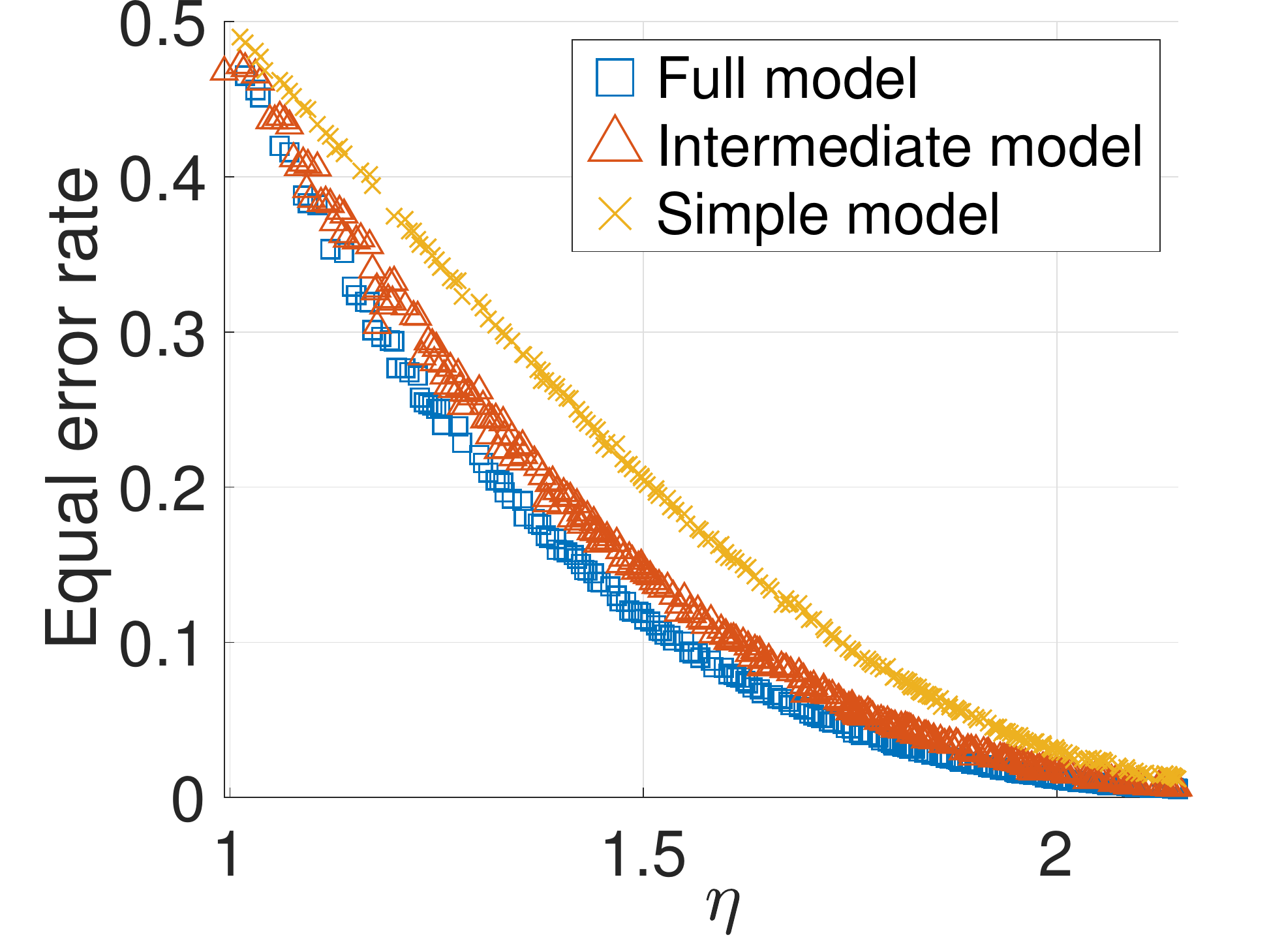}
		\caption{Pareto frontiers of the full, intermediate \& simplified models ($m=6, W=1000$).}\label{fig:paretoEfficiency}
	\end{minipage}
	\hfill
	\begin{minipage}[b]{0.48\textwidth}
		\centering
		\includegraphics[scale=0.28]{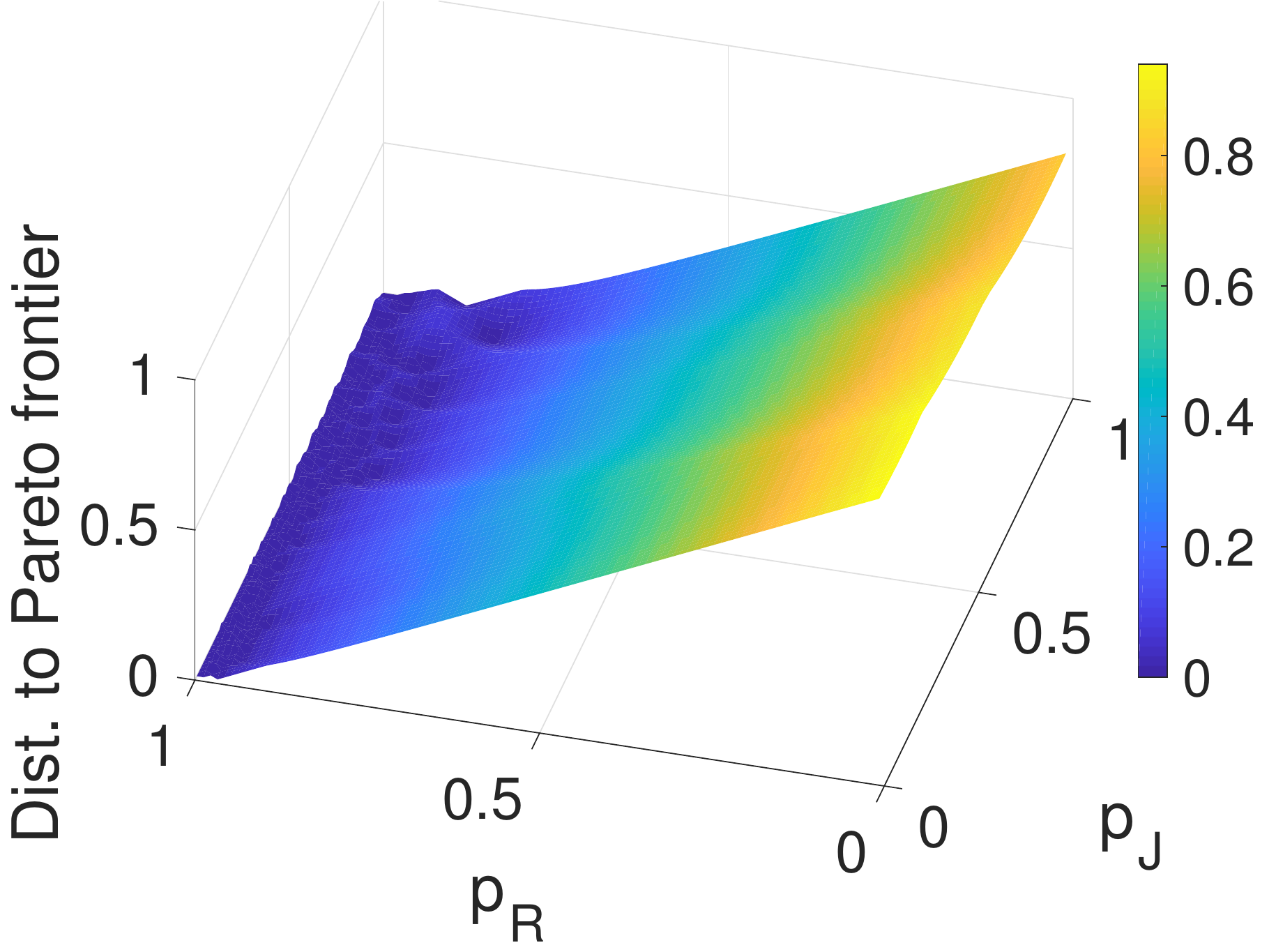}
		\caption{Distance to the Pareto frontier on the $p_J$-$p_R$ plane (full model, $m=6, W=1000$).}\label{fig:distToPf}
	\end{minipage}
	\begin{minipage}[b]{0.36\textwidth}
		\centering
		\subfloat[ROC curve]{\includegraphics[scale=0.25]{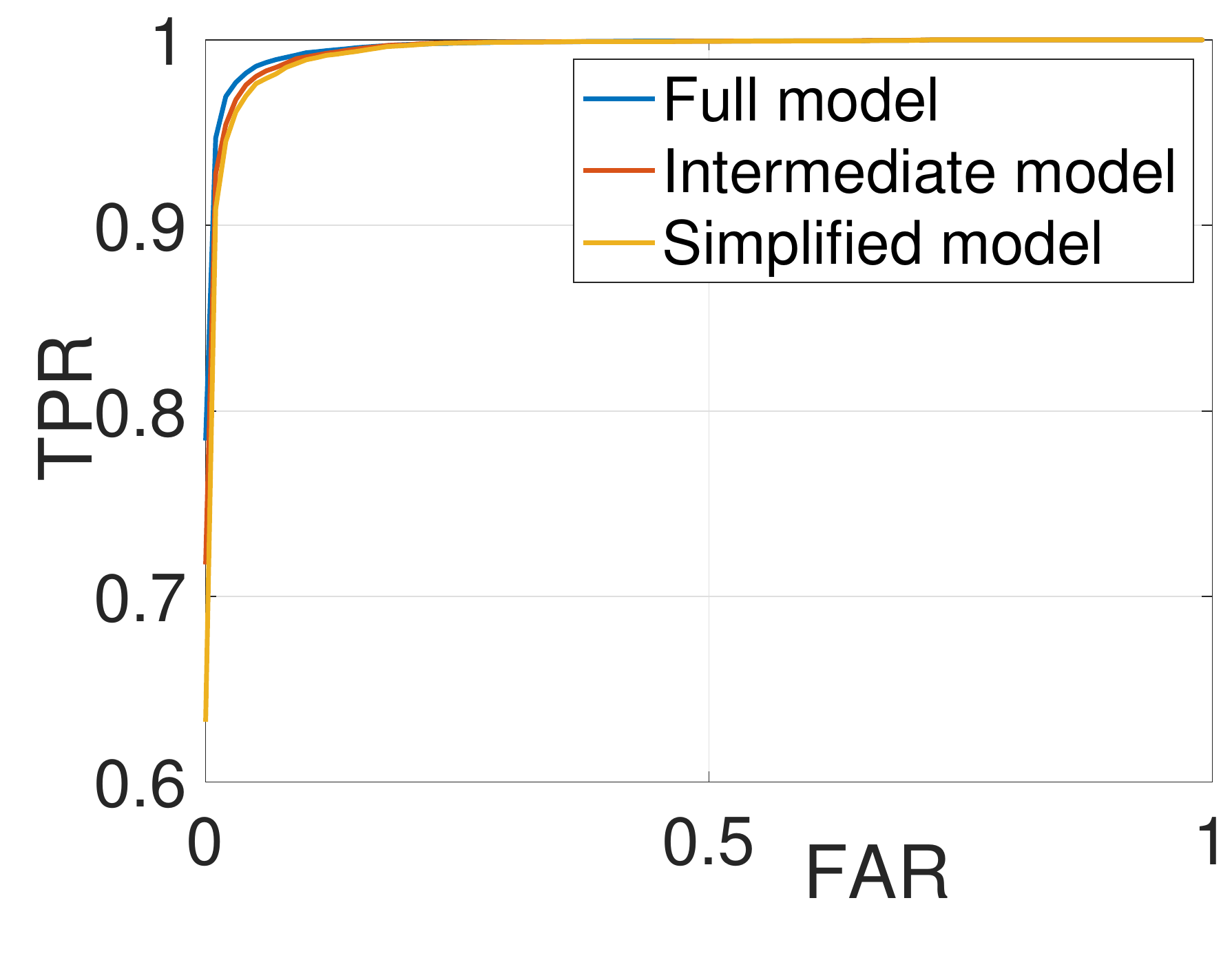}}
		\hfill
		\subfloat[Location plot]{\includegraphics[scale=0.25]{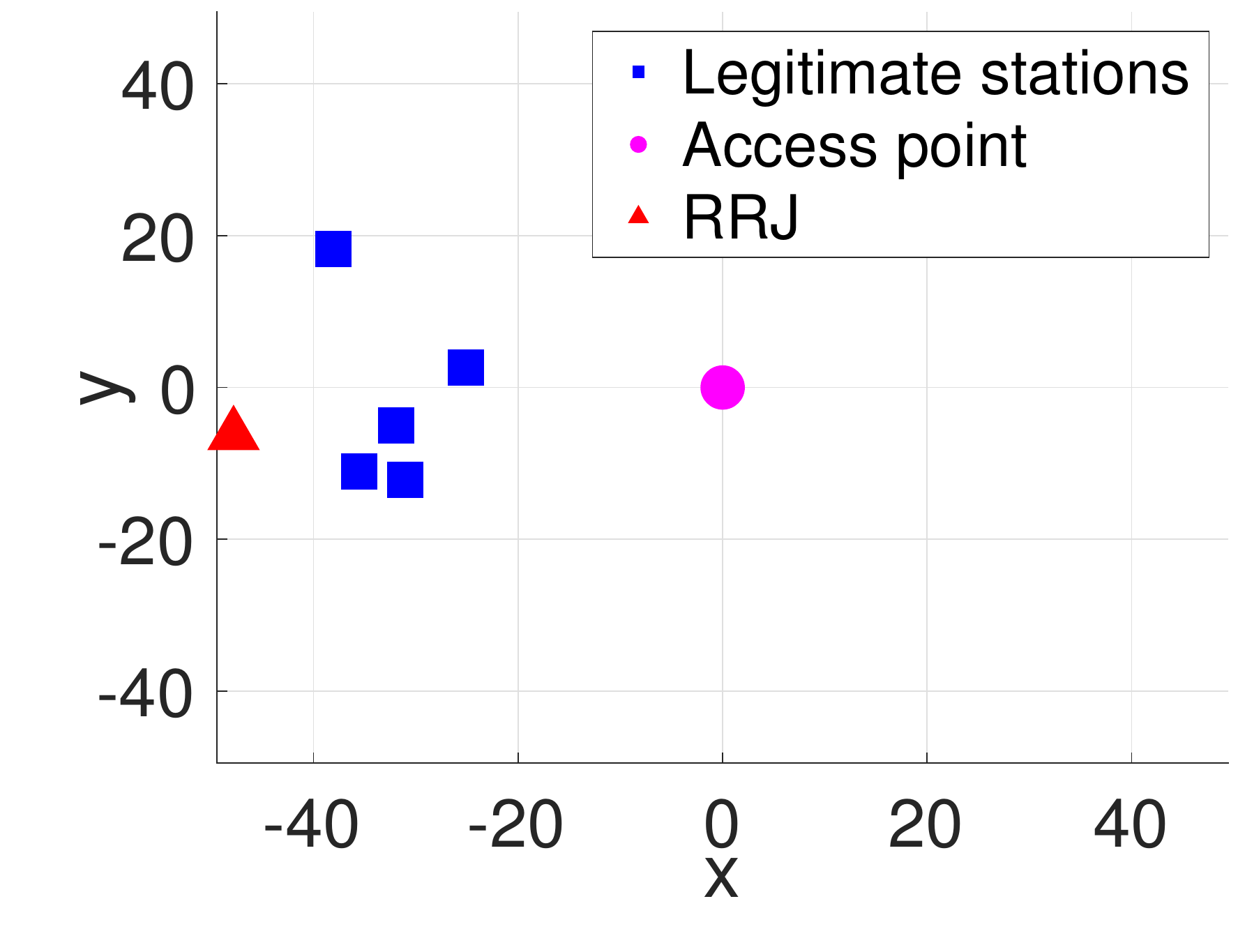}}
	\end{minipage}
	\hfill
	\begin{minipage}[b]{0.6\textwidth}
		\centering
		\subfloat[Case one: ROC curve\label{fig:556}]{\includegraphics[scale=0.25]{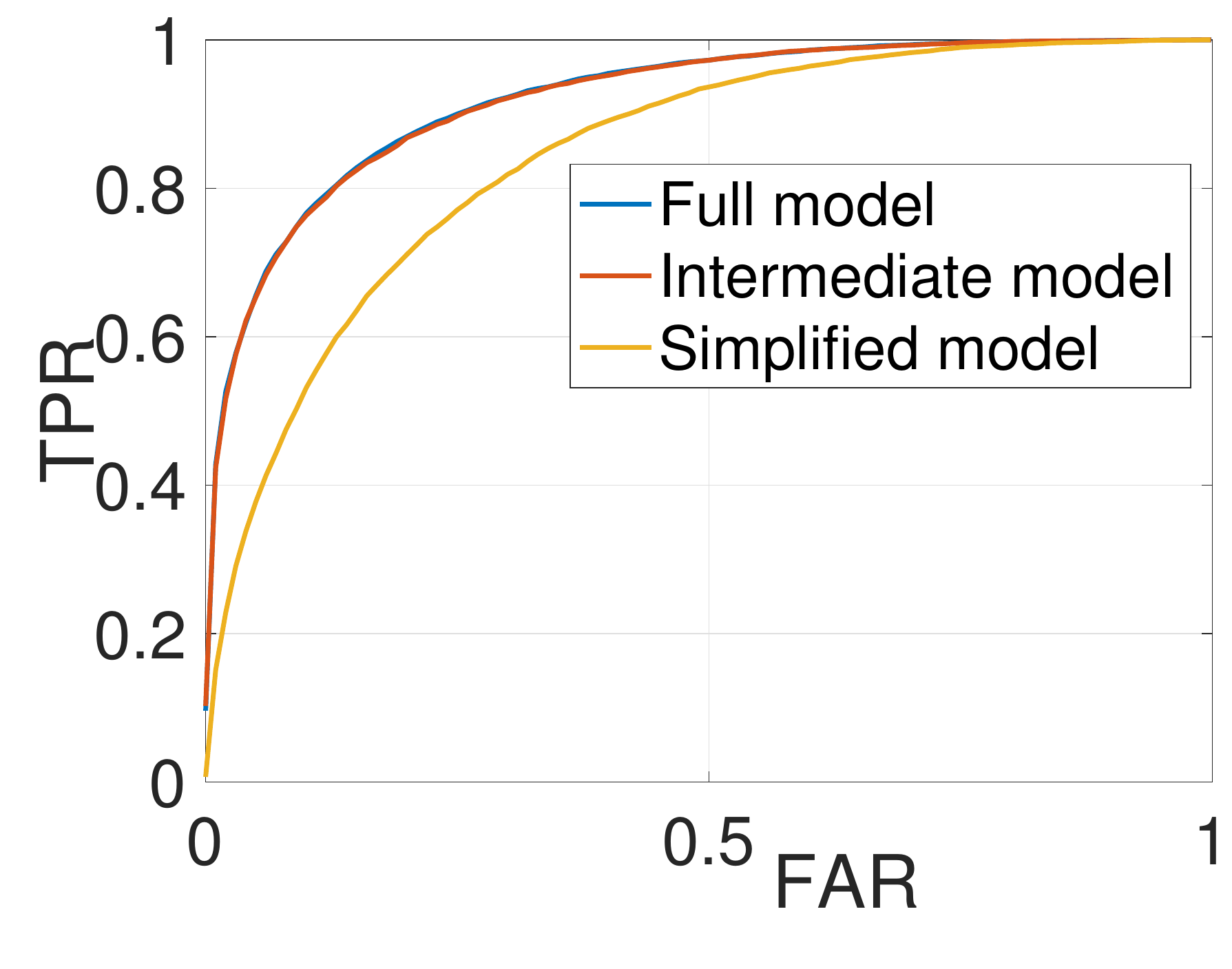}}
		\hfill
		\subfloat[Case two: ROC curve\label{fig:41}]{\includegraphics[scale=0.25]{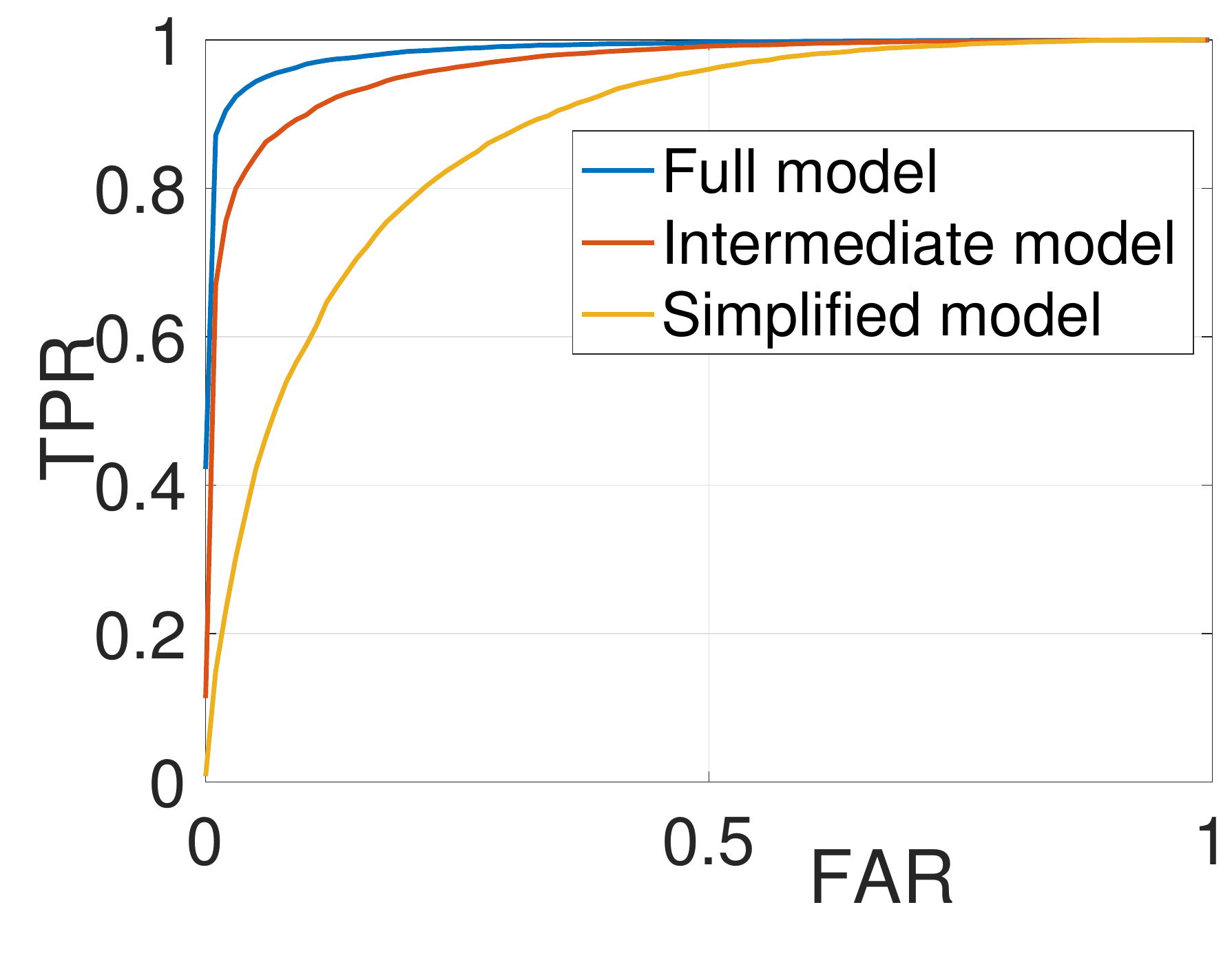}}
		\hfill
		\subfloat[Case one: Location plot]{\includegraphics[scale=0.25]{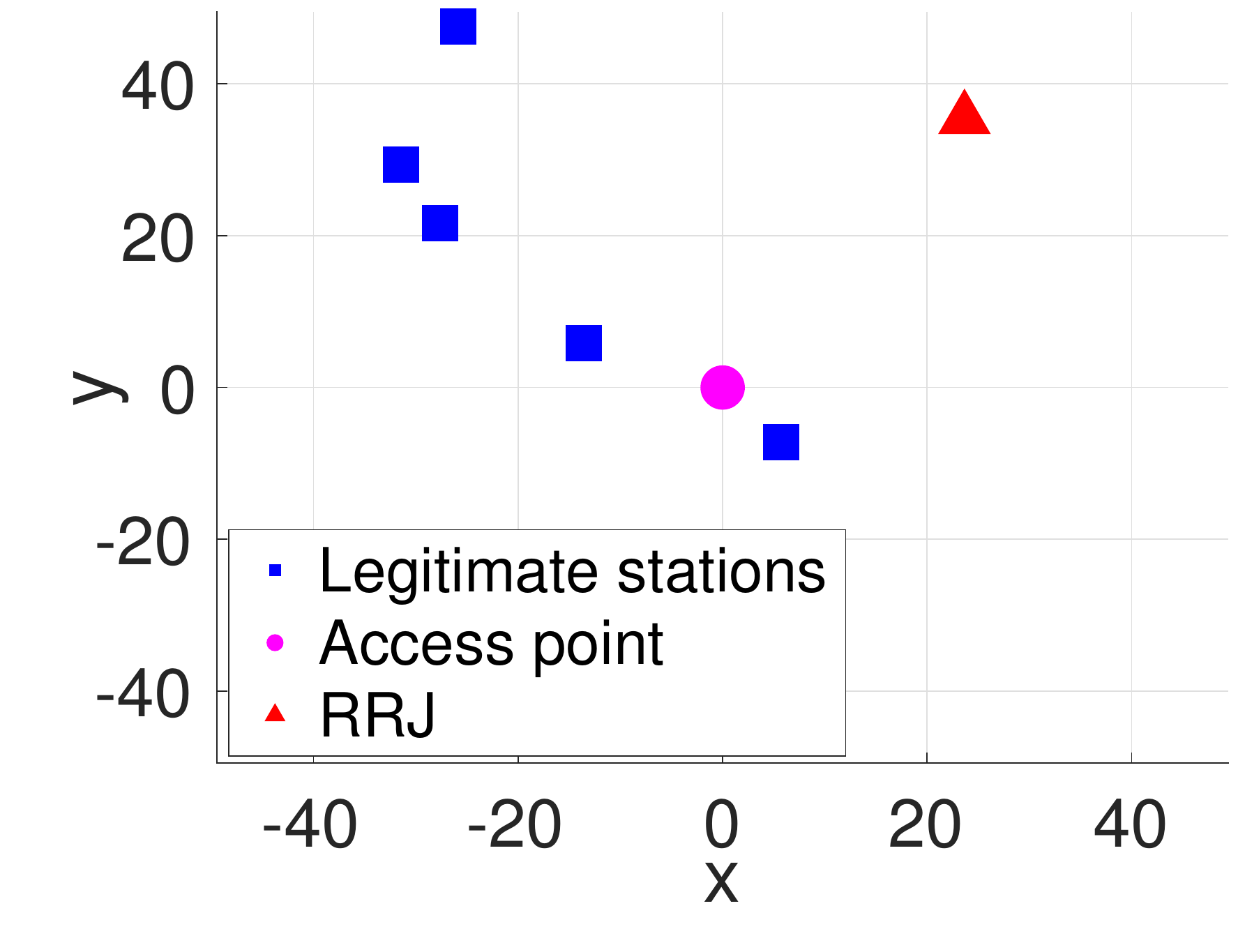}}
		\hfill
		\subfloat[Case two: Location plot]{\includegraphics[scale=0.25]{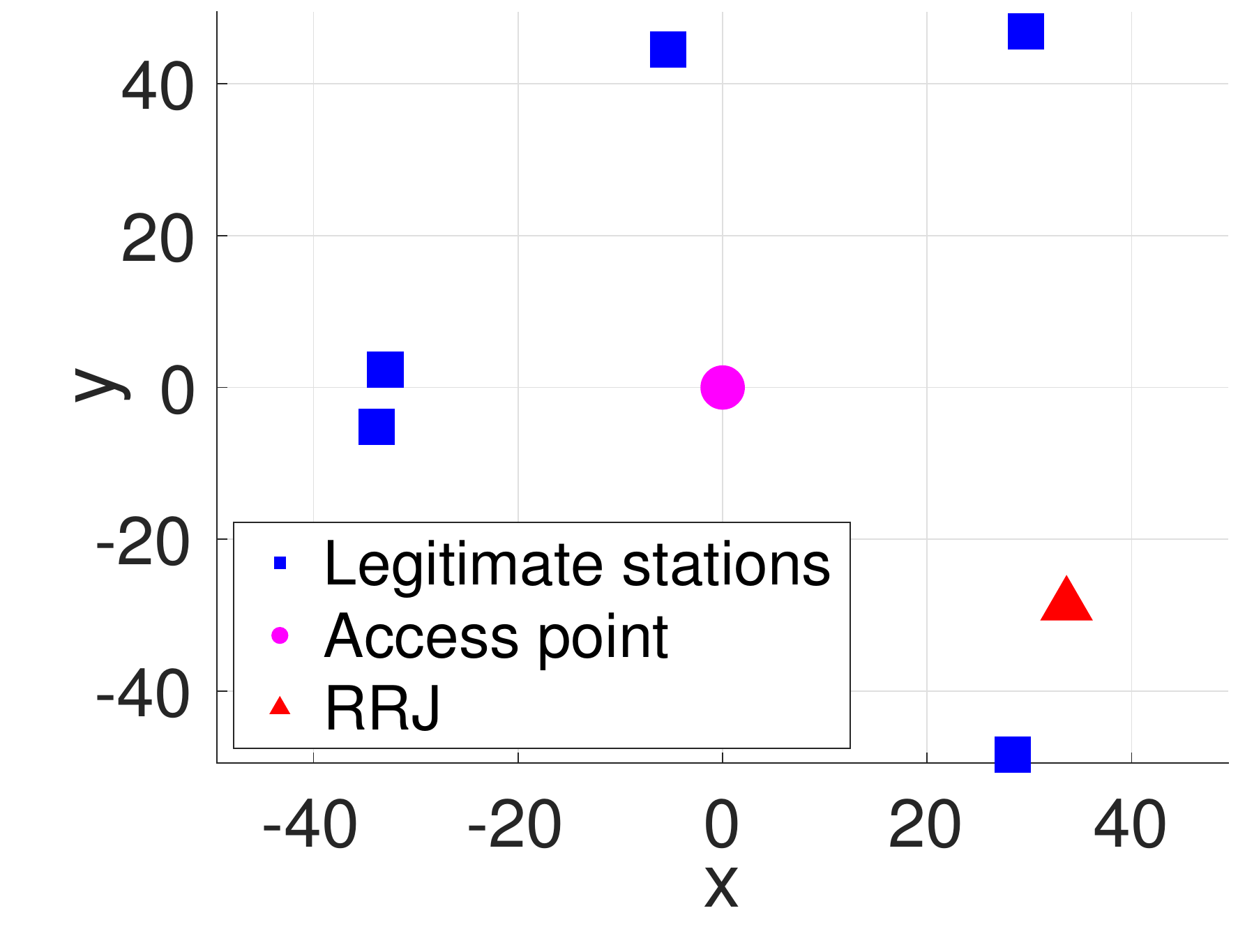}}
	\end{minipage}
	\begin{minipage}[t]{.38\linewidth}
		\caption{ROC and location graph when all three models have similar detection accuracy ($\pbf_{\rm rrj}=[0.8, 0.2]$).}\label{fig:rocFullVsLimited-1}
	\end{minipage}%
	\hfill
	\begin{minipage}[t]{.58\linewidth}
		\caption{ROC and location graph for two example cases that the simplified model has significantly higher EER than the intermediate and full model ($\pbf_{\rm rrj}=[0.8, 0.2]$).}\label{fig:rocFullVsLimited-2}
	\end{minipage}%
\end{figure}

\subsection{Comparison between the full and limited observability models}
This section compares the supervised hypothesis testing of the full observability and the two limited observability models. In the simulation, we generate $n=10^4$ compliant sample paths $\ybf^{\zerosf}_W$ (correspond to the DTMC with $\Pbf^{\zerosf}$) and $n$ RRJ sample paths $\ybf^{\onesf}_W$ (correspond to the DTMC with $\Pbf^{\onesf}$). For each sample path with length $W$, we can compute the transition counts $N_{i,j}$ of the full model, and also we can compute the aggregated transition counts $\hat{N}_{i,j}$ of the intermediate model, and $\tilde{N}_{i,j}$ of the simplified model. With the transition counts, we can compute the test statistic $Z$ of each sample path and perform the supervised hypothesis testing to determine if a sample path is generated from the compliant DTMC or the RRJ DTMC. \figref{paretoEfficiency} shows the Pareto efficient allocations of the full, intermediate and simplified models with the topology in \figref{systemTopo}. The Pareto efficiency is defined from the point of view of RRJ, {\em i.e.}, a RRJ will prefer large EER and large $\eta$. \figref{paretoEfficiency} demonstrates that the simplified model is advantageous for RRJ in that it has a higher EER and $\eta$. This is due to the fact that the amount of information available to the jamming detector is significantly reduced in the four-state simplified model than the other two models with state space size $2m$ and $2^m$. \figref{distToPf} shows the distance of each operating point $(p_R, p_J)$ to the Pareto frontier of the full model, and we can see the distance is small with large $p_R$, which demonstrates the Pareto efficient operating points have $p_R$ value close to 1.

To investigate if the large EER gain of the simplified model over the other two models are consistent for various network topologies, we randomize the station location by randomly placing the stations in a $70\times 70$ area, and plot the corresponding ROC of the three models as shown in \figref{rocFullVsLimited-1} and \figref{rocFullVsLimited-2}. We can see from \figref{rocFullVsLimited-1} the full, intermediate, and simplified models all obtain similar detection accuracy. However, \figref{rocFullVsLimited-2} shows two cases that the simplified model has a much larger detection error than the other two models. The difference among the three models' detection accuracy highly depends upon the geographical distributions of stations.

\section{Conclusion}
\label{sec:conclusion}

This paper analyzed the RJ detection problem in the presence of HTs in networks by measuring (non)compliance with CS. Formulating this jamming detection problem as a generic Markov chain hypothesis testing problem, we derived an upper bound on the variance of the test statistic, and developed a novel optimization problem for an intelligent RRJ to determine its optimal operating point at which the desired tradeoff between the jamming efficiency and the risk of exposure is achieved. We introduced three models: the full observability model, the intermediate model, and the simplified model, corresponding to increasingly realistic assumptions regarding the network state information available at the AP. 
\appendices

\section{Proof of \lemref{VarCovNijUpperbound}}
\label{sec:proofVarCovNijUpperbound}

\begin{proof}
For an ergodic CTMC having transition matrix with simple eigenvalues, its transition matrix can be written as: $\Pbf^{\bsf}\in \Rbb^{(d+1)\times (d+1)} =\Ubf \boldsymbol{\Lambda}^{\bsf} \Vbf$ where $\Vbf=[\vbf_0,...,\vbf_d]^T=\Ubf^{-1}$ is the left eigenvector matrix of $\Pbf^{\bsf}$, and $\Ubf=[\ubf_0,...,\ubf_d]$ is the right eigenvector matrix of $\Pbf^{\bsf}$, and $\boldsymbol{\Lambda}=\mbox{diag}\{\lambda_0^{\bsf},...,\lambda_d^{\bsf}\}$ holds the eigenvalues of the matrix sorted according to $|1-\lambda_0^{\bsf}|\leq...\leq |1-\lambda_d^{\bsf}|$. Since $\lambda_0^{\bsf}$ is 1 with left eigenvector $\boldsymbol{\pi}^{\bsf}$ and right eigenvector $\mathbf{1}$, thus we can write 
\begin{equation}
\label{eq:kstepProbRewrite}
[{\Pbf^{\bsf}}^{(t')}]_{i,j} = \left[\sum_{r=0}^d {\lambda_r^{\bsf}}^{(t')} \ubf_r\vbf_r^{\intercal}\right]_{i,j}=\pi_j^{\bsf}+ \sum_{r=1}^d {\lambda_r^{\bsf}}^{(t')}u_{ir}v_{rj},
\end{equation}
where ${\lambda_r^{\bsf}}^{(t')}$ denotes the $t'$-th power of ${\lambda_r^{\bsf}}$. First, we want to derive the upper bound of ${\rm Var}[N_{i,j}]$ in \eqref{VarNij}. According to \eqref{kstepProbRewrite}, we have
\begin{equation}
\label{eq:VarNijPartRewrite}
\begin{aligned}
\sum_{t'=1 }^{W}(W-t'){\epsilon_{j,i}^{\bsf}}^{(t'-1)} 
= & \sum_{t'=1 }^{W}(W-t')\sum_{r=1}^d {\lambda_r^{\bsf}}^{(t'-1)}u_{jr}v_{ri}
=  \sum_{r=1}^d u_{jr}v_{ri}\sum_{t'=1 }^{W}(W-t'){\lambda_r^{\bsf}}^{(t'-1)}\\
= & \sum_{r=1}^d u_{jr}v_{ri}\frac{{\lambda_r^{\bsf}}^{(W)}-W{\lambda_r^{\bsf}}+(W-1)}{(1-{\lambda_r^{\bsf}})^2}
=  \sum_{r=1}^d u_{jr}v_{ri}\left(\frac{{\lambda_r^{\bsf}}^{(W)}-1}{(1-{\lambda_r^{\bsf}})^2} + \frac{W}{1-{\lambda_r^{\bsf}}}\right)\\
\overset{(a)}{\leq} &  \sum_{r=1}^d |u_{jr}v_{ri}|\left(\frac{2}{|1-{\lambda_r^{\bsf}}|^2} + \frac{W}{|1-{\lambda_r^{\bsf}}|}\right)\\
\overset{(b)}{\leq} &  \sum_{r=1}^d |u_{jr} v_{ri}|\left(\frac{2}{|1-\lambda_1^{\bsf}|^2} + \frac{W}{|1-\lambda_1^{\bsf}|}\right)
\leq  c_{j,i}\frac{2+W|1-\lambda_1^{\bsf}|}{|1-\lambda_1^{\bsf}|^2}
\end{aligned}
\end{equation}
where $(a)$ follows by the triangle inequality of complex numbers and $|{\lambda_r^{\bsf}}|\leq 1\mbox{ for } r\geq 1$, $(b)$ follows by $|1-\lambda_1|\leq |1-{\lambda_r^{\bsf}}|$ for $r\geq 2$, $c_{j,i}\equiv \sum_{r=1}^d |u_{jr}v_{ri}|$ is a constant depends on the eigenvectors of $\Pbf^{\bsf}$.  According to \eqref{VarNij}, and \eqref{VarNijPartRewrite}, we can then derive \eqref{VarNijUpperbound}. Similarly,
\begin{equation}
\label{eq:CovNijPartUpperbound}
-W\pi_i^{\bsf} + 2\sum_{t'=1}^{W-1}(W-t')\left({\epsilon_{j,i}^{\bsf}}^{(t'-1)}+{\epsilon_{j',i}^{\bsf}}^{(t'-1)}\right)
\leq -W{\pi_i^{\bsf}} + (c_{j,i}+c_{j',i})\frac{ (2+W|1-\lambda_1^{\bsf}|)}{|1-\lambda_1^{\bsf}|^2}
\end{equation}
The upper-bound of ${\rm Cov}[\frac{N_{i,j}^{\bsf}}{W},\frac{N_{i,j'}^{\bsf}}{W}]$ follows immediately by substituting \eqref{CovNijPartUpperbound} into \eqref{CovNij}.
\end{proof}

\section{Proof of \lemref{convergenceRateConvex}}
\label{sec:proofConvergenceRateConvex}

\begin{proof}
Define the transition probability vector of full model from state $i$ to states in $\Smc$ as: $\pbf_i^{\zerosf}\in \Rbb^{1\times (d+1)}$, $\pbf_i^{\onesf}\in \Rbb^{1\times (d+1)}$. 
\begin{equation}
\label{eq:informationDiscrimination}
\begin{aligned}
I(\Pbf^{\zerosf},\Pbf^{\onesf}(\pbf_{\rm rrj})) = \sum_{i}\pi_i^{\zerosf} \sum_{j} p_{i,j}^{\zerosf}\ln\frac{p_{i,j}^{\zerosf}}{p_{i,j}^{\onesf}}
\stackrel{(a)}{=} \sum_{i^*\in \Smc_{\neg 1}}\pi_{i^*}^{\zerosf} D(\pbf_{i^*}^{\zerosf}||\pbf_{i^*}^{\onesf}),
\end{aligned}
\end{equation}
where $\Smc_{\neg 1}\equiv \{\Tmc\in\Smc: 1\notin \Tmc\}$ denotes the set of the indices of all states that do not contain the SUT, $D(\pbf||\qbf)$ denotes the Kullback-Leibler (KL) divergence between $\pbf$ and $\qbf$; $(a)$ is due to the fact that $D(\pbf||\qbf)$ is non-zero when $\pbf\neq \qbf$, and the transition probabilities ${\pbf_{i^*}}^{\zerosf}$ and ${\pbf_{i^*}}^{\onesf}$ are not equal only when $i^*\in \Smc_{\neg 1}$ according to the definition of RRJ's Markov chain in \secref{generalFullModel}. It is known that the KL divergence $D(\pbf||\qbf)$ is a convex function in $\qbf$ \cite{CovTho2012}. Since in our case, $\pbf_{i^*}^{\onesf}$ is an affine function of the jamming parameters of the RRJ $\pbf_{i^*}^{\onesf} =  \Abf [p_R, p_J]^{\intercal} +\bbf$, where $\Abf$ is a matrix with its $h_{j^*}$-th row ($h_{j^*}$ is the index of state $j^*=\Tmc\cup \{1\}$) being: $[\frac{\lambda}{u} p_I(1,\Tmc),\frac{\lambda}{u}( 1- p_I(1,\Tmc))]$, and all the other elements are zeros, and $\bbf$ is an array with each element $b_j$ being: $b_j = p_{i^*,j}^{\zerosf}$ if $j\neq j^*$; $b_j =0$ if $j=j^*$. Since affine mapping preserves the convexity of the function \cite{BoyVan2004},  $D({\pbf_{i^*}}^{\zerosf}||\pbf_{i^*}^{\onesf})$ is a convex function of $[p_R, p_J]$. Therefore, as a non-negative weighted sum of convex functions of $[p_R, p_J]$ as in \eqref{informationDiscrimination}, 	$I(\Pbf^{\zerosf},\Pbf^{\onesf}(\pbf_{\rm rrj}))$ is also a convex function of $[p_R, p_J]$.
\end{proof}

\bibliographystyle{IEEEtran}
\bibliography{IEEEabrv,TCOM}

\end{document}